\documentclass[a4paper,UKenglish,cleveref, autoref, thm-restate,authorcolumns]{lipics-v2019}
\pdfoutput=1 
\usepackage[utf8]{inputenc}
\usepackage[T1]{fontenc}
\usepackage{xcolor}
\usepackage{amsmath}
\usepackage{amsthm}
\usepackage{amssymb}
\usepackage{mathcommand}
\usepackage{cite}
\usepackage{tikz}
\usepackage{floatflt}
\usepackage{mathtools}
\usepackage{wrapfig}
\usetikzlibrary{automata, positioning, arrows, cd}
\usetikzlibrary{decorations.markings}
\tikzset{negated/.style={
		decoration={markings,
			mark= at position 0.5 with {
				\node[transform shape] (tempnode) {$\slash$};
			}
		},
		postaction={decorate}
	}
}
\usepackage[notion,quotation,protect quotation={tikzcd},electronic]{knowledge}
\knowledgeconfigure{diagnose line=true}

\usepackage[linesnumbered,lined,boxed,commentsnumbered]{algorithm2e}

 \usepackage{enumitem}
 \setlist[itemize]{wide=0pt, nosep}

\newmathcommand\catSet{\mathsf{Set}}
\newmathcommand\catVec{\mathsf{Vec}}
\newmathcommand\K{\kl[\K]{\mathbb{K}}}
\knowledge\K{notion}
  
\newmathcommand\catKl{\mathsf{Kl}}

\newmathcommand\monT{\kl[\monT]{\mathcal{T}}}
\knowledge\monT{notion}

\newmathcommand\catC{\kl[\catC]{\mathcal{C}}}
\knowledge\catC{notion}
\newmathcommand\catCintro{\kl[\catCintro]{\mathcal{C}}}
\knowledge\catCintro{notion}

\newmathcommand\catK{\kl[\catK]{\mathcal{K}}}
\knowledge\catK{notion}

\newmathcommand\catTargC{\kl[\catTargC]{\mathcal{C}}}
\knowledge\catTargC{notion}

\newmathcommand\catI{\kl[\catI]{\mathcal{I}_{\alphA^*}}}
\knowledge\catI{notion} % Word automata case

\newmathcommand\catIstart{\kl[\catIstart]{\mathcal{I}}}
\knowledge\catIstart{notion} % General case
\newmathcommand\catIintro{\kl[\catIintro]{\mathcal{I}}}
\knowledge\catIintro{notion} % General case

\newmathcommand\catO{\kl[\catO]{\mathcal{O}_{A^*}}}
\knowledge\catO{notion} % Word automata case

\newmathcommand\catOstart{\kl[\catOstart]{\mathcal{O}}}
\knowledge\catOstart{notion} % General case

\newmathcommand\catIQT[1][\setQ,\setT]{\kl[\catIQT]{\mathcal{I}}_{#1}}
\knowledge\catIQT{notion}

\newmathcommand\catOQT[1][\setQ,\setT]{\kl[\catOQT]{\mathcal{O}}_{#1}}
\knowledge\catOQT{notion}

\newmathcommand\langL{\kl[\langL]{\mathcal{L}}}
\knowledge\langL{notion}
\newmathcommand\targL{\kl[\targL]{\mathcal{L}}}
\knowledge\targL{notion}
\newmathcommand\iStar{\kl[\iStar]{i^*}}
\knowledge\iStar{notion}
\newmathcommand\langLQT[1][\setQ,\setT]{\kl[\langLQT]{\mathcal{L}_{#1}}}
\knowledge\langLQT{notion}

\newmathcommand\catAuto{\kl[\catAuto]{\mathsf{Auto}}}
\knowledge\catAuto{notion}

\newmathcommand\setQintro{\kl[\setQintro]{Q}}
\knowledge\setQintro{notion}
\newmathcommand\setTintro{\kl[\setTintro]{T}}
\knowledge\setTintro{notion}
\newmathcommand\setQ{\kl[\setQ]{Q}}
\knowledge\setQ{notion}
\newmathcommand\setT{\kl[\setT]{T}}
\knowledge\setT{notion}
\newmathcommand\setQpr{\kl[\setQpr]{Q'}}
\knowledge\setQpr{notion}
\newmathcommand\setTpr{\kl[\setTpr]{T'}}
\knowledge\setTpr{notion}
\newmathcommand\setQafter{\kl[\setQafter]{Q'}}
\knowledge\setQafter{notion}
\newmathcommand\setTafter{\kl[\setTafter]{T'}}
\knowledge\setTafter{notion}
\newmathcommand\res[1][T',T]{\kl[\res]{\mathit{res_{#1}}}}
\knowledge\res{notion}
\newmathcommand\inc[1][Q, Q']{\kl[\inc]{\mathit{inc_{#1}}}}
\knowledge\inc{notion}

\newmathcommand\catAutoQT[1][\setQ,\setT]{\kl[\catAutoQT]{\mathsf{Auto}_{#1}}}
\knowledge\catAutoQT{notion}
\newmathcommand\factQT[1][\setQ,\setT]{\kl[\factQT]{\Im_{#1}}}
\knowledge\factQT{notion}
\newmathcommand\fact[1]{\kl[\factQT]{\Im_{#1}}}
\knowledge\fact{notion}

\newmathcommand\eWhile{\kl[\eWhile]{e}}
\knowledge\eWhile{notion}
\newmathcommand\mWhile{\kl[\mWhile]{m}}
\knowledge\mWhile{notion}
\newmathcommand\fLem{\kl[\fLem]{f}}
\knowledge\fLem{notion}
\newmathcommand\eLem{\kl[\eLem]{e}}
\knowledge\eLem{notion}
\newmathcommand\mLem{\kl[\mLem]{m}}
\knowledge\mLem{notion}
\newmathcommand\fLemOv{\kl[\fLemOv]{\overline{f}}}
\knowledge\fLemOv{notion}
\newmathcommand\eLemOv{\kl[\eLemOv]{\overline{e}}}
\knowledge\eLemOv{notion}
\newmathcommand\mLemOv{\kl[\mLemOv]{\overline{m}}}
\knowledge\mLemOv{notion}
\newmathcommand\fLemUn{\kl[\fLemUn]{\underline{f}}}
\knowledge\fLemUn{notion}
\newmathcommand\eLemUn{\kl[\eLemUn]{\underline{e}}}
\knowledge\eLemUn{notion}
\newmathcommand\mLemUn{\kl[\mLemUn]{\underline{m}}}
\knowledge\mLemUn{notion}
\newmathcommand\eWhileOv{\kl[\eWhileOv]{\overline{e}}}
\knowledge\eWhileOv{notion}
\newmathcommand\mWhileOv{\kl[\mWhileOv]{\overline{m}}}
\knowledge\mWhileOv{notion}
\newmathcommand\eWhileUn{\kl[\eWhileUn]{\underline{e}}}
\knowledge\eWhileUn{notion}
\newmathcommand\mWhileUn{\kl[\mWhileUn]{\underline{m}}}
\knowledge\mWhileUn{notion}

\newrobustcmd\FunL{\ensuremath{\kl[\FunL]{\mathtt{FunL\mkern-1.5mu}^*}\mkern-.5mu}}
\knowledge\FunL{notion}
\newrobustcmd\Lstar{\mathtt{L\mkern-1.5mu}^*\mkern-.5mu}

\newrobustcmd\symbLeft{\kl[\symbLeft]\triangleright}
\knowledge\symbLeft{notion}
\newrobustcmd\symbRight{\kl[\symbRight]\triangleleft}
\knowledge\symbRight{notion}
\newrobustcmd\symbA{\kl[\symbA]a}
\knowledge\symbA{notion}

\newrobustcmd\symbLeftQ[1][q]{\kl[\symbLeftQ]{\triangleright #1}}
\knowledge\symbLeftQ{notion}
\newrobustcmd\symbRightT[1][t]{\kl[\symbRightT]{#1 \triangleleft}}
\knowledge\symbRightT{notion}
\newrobustcmd\symbEpsilon{\kl[\symbEpsilon]\varepsilon}
\knowledge\symbEpsilon{notion}
\newrobustcmd\symbAQT[1][a]{\kl[\symbAQT]{#1}}
\knowledge\symbAQT{notion}

\newrobustcmd\symbLeftQInit{\kl[\symbLeftQInit]{\triangleright q_{\mathit{init}}}}
\knowledge\symbLeftQInit{notion}
\newrobustcmd\symbRightTInit{\kl[\symbRightTInit]{t\triangleleft_{\mathit{init}}}}
\knowledge\symbRightTInit{notion}
\newrobustcmd\symbEpsilonInit{\kl[\symbEpsilonInit]{\varepsilon_{\mathit{init}}}}
\knowledge\symbEpsilonInit{notion}
\newrobustcmd\symbAQTInit{\kl[\symbAQTInit]{a_{\mathit{init}}}}
\knowledge\symbAQTInit{notion}

\newrobustcmd\symbLeftQFinal{\kl[\symbLeftQFinal]{\triangleright q_{\mathit{final}}}}
\knowledge\symbLeftQFinal{notion}
\newrobustcmd\symbRightTFinal{\kl[\symbRightTFinal]{t\triangleleft_{\mathit{final}}}}
\knowledge\symbRightTFinal{notion}
\newrobustcmd\symbEpsilonFinal{\kl[\symbEpsilonFinal]{\varepsilon_{\mathit{final}}}}
\knowledge\symbEpsilonFinal{notion}
\newrobustcmd\symbAQTFinal{\kl[\symbAQTFinal]{a_{\mathit{final}}}}
\knowledge\symbAQTFinal{notion}

\newrobustcmd\symbLeftQMin[1][q]{\kl[\symbLeftQMin]{\triangleright #1_{\mathit{min}}}}
\knowledge\symbLeftQMin{notion}
\newrobustcmd\symbRightTMin[1][t]{\kl[\symbRightTMin]{#1\triangleleft_{\mathit{min}}}}
\knowledge\symbRightTMin{notion}
\newrobustcmd\symbEpsilonMin{\kl[\symbEpsilonMin]{\varepsilon_{\mathit{min}}}}
\knowledge\symbEpsilonMin{notion}

\newmathcommand\symbEpsilonMinQT[1][\setQ,\setT]{\kl[\symbEpsilonMinQT]{\varepsilon_{\mathit{min}}^{#1}}}
\knowledge\symbEpsilonMinQT{notion}

\newrobustcmd\symbAQTMin[1][a]{\kl[\symbAQTMin]{#1_{\mathit{min}}}}
\knowledge\symbAQTMin{notion}

\newrobustcmd\symbLeftQInitPr{\kl[\symbLeftQInitPr]{\triangleright q_{\mathit{init'}}}}
\knowledge\symbLeftQInitPr{notion}
\newrobustcmd\symbRightTInitPr{\kl[\symbRightTInitPr]{t\triangleleft_{\mathit{init'}}}}
\knowledge\symbRightTInitPr{notion}
\newrobustcmd\symbEpsilonInitPr{\kl[\symbEpsilonInitPr]{\varepsilon_{\mathit{init'}}}}
\knowledge\symbEpsilonInitPr{notion}
\newrobustcmd\symbAQTInitPr{\kl[\symbAQTInitPr]{a_{\mathit{init'}}}}
\knowledge\symbAQTInitPr{notion}

\newrobustcmd\symbLeftQInitDub{\kl[\symbLeftQInitDub]{\triangleright q_{\mathit{init'}}}}
\knowledge\symbLeftQInitDub{notion}
\newrobustcmd\symbRightTInitDub{\kl[\symbRightTInitDub]{t\triangleleft_{\mathit{init'}}}}
\knowledge\symbRightTInitDub{notion}
\newrobustcmd\symbEpsilonInitDub{\kl[\symbEpsilonInitDub]{\varepsilon_{\mathit{init'}}}}
\knowledge\symbEpsilonInitDub{notion}
\newrobustcmd\symbAQTInitDub{\kl[\symbAQTInitDub]{a_{\mathit{init'}}}}
\knowledge\symbAQTInitDub{notion}

\newrobustcmd\symbLeftQInitTr{\kl[\symbLeftQInitTr]{\triangleright q_{\mathit{init'}}}}
\knowledge\symbLeftQInitTr{notion}
\newrobustcmd\symbRightTInitTr{\kl[\symbRightTInitTr]{t\triangleleft_{\mathit{init'}}}}
\knowledge\symbRightTInitTr{notion}
\newrobustcmd\symbEpsilonInitTr{\kl[\symbEpsilonInitTr]{\varepsilon_{\mathit{init'}}}}
\knowledge\symbEpsilonInitTr{notion}
\newrobustcmd\symbAQTInitTr{\kl[\symbAQTInitTr]{a_{\mathit{init'}}}}
\knowledge\symbAQTInitTr{notion}

\newrobustcmd\symbLeftQFinalPr{\kl[\symbLeftQFinalPr]{\triangleright q_{\mathit{final'}}}}
\knowledge\symbLeftQFinalPr{notion}
\newrobustcmd\symbRightTFinalPr{\kl[\symbRightTFinalPr]{t\triangleleft_{\mathit{final'}}}}
\knowledge\symbRightTFinalPr{notion}
\newrobustcmd\symbEpsilonFinalPr{\kl[\symbEpsilonFinalPr]{\varepsilon_{\mathit{final'}}}}
\knowledge\symbEpsilonFinalPr{notion}
\newrobustcmd\symbAQTFinalPr{\kl[\symbAQTFinalPr]{a_{\mathit{final'}}}}
\knowledge\symbAQTFinalPr{notion}

\newrobustcmd\symbLeftQFinalDub{\kl[\symbLeftQFinalDub]{\triangleright q_{\mathit{final'}}}}
\knowledge\symbLeftQFinalDub{notion}
\newrobustcmd\symbRightTFinalDub{\kl[\symbRightTFinalDub]{t\triangleleft_{\mathit{final'}}}}
\knowledge\symbRightTFinalDub{notion}
\newrobustcmd\symbEpsilonFinalDub{\kl[\symbEpsilonFinalDub]{\varepsilon_{\mathit{final'}}}}
\knowledge\symbEpsilonFinalDub{notion}
\newrobustcmd\symbAQTFinalDub{\kl[\symbAQTFinalDub]{a_{\mathit{final'}}}}
\knowledge\symbAQTFinalDub{notion}

\newrobustcmd\symbLeftQFinalTr{\kl[\symbLeftQFinalTr]{\triangleright q_{\mathit{final'}}}}
\knowledge\symbLeftQFinalTr{notion}
\newrobustcmd\symbRightTFinalTr{\kl[\symbRightTFinalTr]{t\triangleleft_{\mathit{final'}}}}
\knowledge\symbRightTFinalTr{notion}
\newrobustcmd\symbEpsilonFinalTr{\kl[\symbEpsilonFinalTr]{\varepsilon_{\mathit{final'}}}}
\knowledge\symbEpsilonFinalTr{notion}
\newrobustcmd\symbAQTFinalTr{\kl[\symbAQTFinalTr]{a_{\mathit{final'}}}}
\knowledge\symbAQTFinalTr{notion}

\newrobustcmd\symbLeftQMinPr[1][q]{\kl[\symbLeftQMinPr]{\triangleright #1_{\mathit{min'}}}}
\knowledge\symbLeftQMinPr{notion}
\newrobustcmd\symbRightTMinPr[1][t]{\kl[\symbRightTMinPr]{#1\triangleleft_{\mathit{min'}}}}
\knowledge\symbRightTMinPr{notion}
\newrobustcmd\symbEpsilonMinPr{\kl[\symbEpsilonMinPr]{\varepsilon_{\mathit{min'}}}}
\knowledge\symbEpsilonMinPr{notion}
\newrobustcmd\symbAQTMinPr[1][a]{\kl[\symbAQTMinPr]{#1_{\mathit{min'}}}}
\knowledge\symbAQTMinPr{notion}

\newrobustcmd\symbLeftQMinDub[1][q]{\kl[\symbLeftQMinDub]{\triangleright #1_{\mathit{min'}}}}
\knowledge\symbLeftQMinDub{notion}
\newrobustcmd\symbRightTMinDub[1][t]{\kl[\symbRightTMinDub]{#1\triangleleft_{\mathit{min'}}}}
\knowledge\symbRightTMinDub{notion}
\newrobustcmd\symbEpsilonMinDub{\kl[\symbEpsilonMinDub]{\varepsilon_{\mathit{min'}}}}
\knowledge\symbEpsilonMinDub{notion}
\newrobustcmd\symbAQTMinDub[1][a]{\kl[\symbAQTMinDub]{#1_{\mathit{min'}}}}
\knowledge\symbAQTMinDub{notion}

\newrobustcmd\symbLeftQMinTr[1][q]{\kl[\symbLeftQMinTr]{\triangleright #1_{\mathit{min'}}}}
\knowledge\symbLeftQMinTr{notion}
\newrobustcmd\symbRightTMinTr[1][t]{\kl[\symbRightTMinTr]{#1\triangleleft_{\mathit{min'}}}}
\knowledge\symbRightTMinTr{notion}
\newrobustcmd\symbEpsilonMinTr{\kl[\symbEpsilonMinTr]{\varepsilon_{\mathit{min'}}}}
\knowledge\symbEpsilonMinTr{notion}
\newrobustcmd\symbAQTMinTr[1][a]{\kl[\symbAQTMinTr]{#1_{\mathit{min'}}}}
\knowledge\symbAQTMinTr{notion}
\newrobustcmd\objIn{\kl[\objIn]{\mathsf{in}}}
\knowledge\objIn{notion}
\newrobustcmd\objOut{\kl[\objOut]{\mathsf{out}}}
\knowledge\objOut{notion}
\newrobustcmd\objMin{\kl[\objMin]{\mathsf{Min}}}
\knowledge\objMin{notion}
\newrobustcmd\objMinStart{\kl[\objMinStart]{\mathsf{Min}}} %Only at the begininning
\knowledge\objMinStart{notion}
\newrobustcmd\objMinim{\kl[\objMinim]{\mathsf{Min}^{\mathit{St}}}}
\knowledge\objMinim{notion}
\newrobustcmd\objMinimQT[1][\setQ,\setT]{\kl[\objMinimQT]{\mathsf{Min}_{#1}}}
\knowledge\objMinimQT{notion}
\newrobustcmd\objMinQT[2][\setQ,\setT]{\kl[\objMinQT#2]{\mathsf{Min}^{#2}_{#1}}}
\knowledge{\objMinQT1}{notion}
\knowledge{\objMinQT2}{notion}
\newrobustcmd\objStates{\kl[\objStates]{\mathsf{st}}}
\knowledge\objStates{notion}
\newmathcommand\epiOne{\kl[\epiOne]{e^1_{\mathit{min}}}}  %Epis related to the minimal
\knowledge\epiOne{notion}
\newmathcommand\epiTwo{\kl[\epiTwo]{e^2_{\mathit{min}}}}
\knowledge\epiTwo{notion}
\newmathcommand\monoOne{\kl[\monoOne]{m^1_{\mathit{min}}}} %Monos related to the minimal
\knowledge\monoOne{notion}
\newmathcommand\monoTwo{\kl[\monoTwo]{m^2_{\mathit{min}}}}
\knowledge\monoTwo{notion}
\newmathcommand\epiOnePr{\kl[\epiOnePr]{e^{1}_{\mathit{min'}}}}
\knowledge\epiOnePr{notion}
\newmathcommand\epiTwoPr{\kl[\epiTwoPr]{e^{2}_{\mathit{min'}}}}
\knowledge\epiTwoPr{notion}
\newmathcommand\monoOnePr{\kl[\monoOnePr]{m^{1}_{\mathit{min'}}}}
\knowledge\monoOnePr{notion}
\newmathcommand\monoTwoPr{\kl[\monoTwoPr]{m^{2}_{\mathit{min'}}}}
\knowledge\monoTwoPr{notion}
\newmathcommand\epiOneDub{\kl[\epiOneDub]{e^{1}_{\mathit{min'}}}}
\knowledge\epiOneDub{notion}
\newmathcommand\epiTwoDub{\kl[\epiTwoDub]{e^{2}_{\mathit{min'}}}}
\knowledge\epiTwoDub{notion}
\newmathcommand\monoOneDub{\kl[\monoOneDub]{m^{2}_{\mathit{min'}}}}
\knowledge\monoOneDub{notion}
\newmathcommand\monoTwoDub{\kl[\monoTwoDub]{m^{2}_{\mathit{min'}}}}
\knowledge\monoTwoDub{notion}
\newmathcommand\epiOneTr{\kl[\epiOneTr]{e^{1}_{\mathit{min'}}}}
\knowledge\epiOneTr{notion}
\newmathcommand\epiTwoTr{\kl[\epiTwoTr]{e^{2}_{\mathit{min'}}}}
\knowledge\epiTwoTr{notion}
\newmathcommand\monoOneTr{\kl[\monoOneTr]{m^{1}_{\mathit{min'}}}}
\knowledge\monoOneTr{notion}
\newmathcommand\monoTwoTr{\kl[\monoTwoTr]{m^{2}_{\mathit{min'}}}}
\knowledge\monoTwoTr{notion}
\newmathcommand\epiGen{\kl[\epiGen]{\widetilde{e}}}
\knowledge\epiGen{notion}
\newmathcommand\monoGen{\kl[\monoGen]{\widetilde{m}}}
\knowledge\monoGen{notion}
\newmathcommand\epiMin{\kl[\epiMin]{e_{\mathit{min}}}}
\knowledge\epiMin{notion}
\newmathcommand\monoMin{\kl[\monoMin]{m_{\mathit{min}}}}
\knowledge\monoMin{notion}
\newmathcommand\epiGenBis{\kl[\epiGenBis]{\widetilde{e}}}
\knowledge\epiGenBis{notion}
\newmathcommand\monoGenBis{\kl[\monoGenBis]{\widetilde{m}}}
\knowledge\monoGenBis{notion}
\newmathcommand\epiGenPr{\kl[\epiGenPr]{e'}}
\knowledge\epiGenPr{notion}
\newmathcommand\monoGenPr{\kl[\monoGenPr]{m'}}
\knowledge\monoGenPr{notion}
\newmathcommand\epiGenDub{\kl[\epiGenDub]{e'}}
\knowledge\epiGenDub{notion}
\newmathcommand\monoGenDub{\kl[\monoGenDub]{m'}}
\knowledge\monoGenDub{notion}

\newrobustcmd\objBiIn{\kl[\objBiIn]{\mathsf{in}}}
\knowledge\objBiIn{notion}
\newrobustcmd\objBiOut{\kl[\objBiOut]{\mathsf{out}}}
\knowledge\objBiOut{notion}
\newrobustcmd\objBiStates[1]{\kl[\objBiStates#1]{\mathsf{st}_{#1}}}
\knowledge{\objBiStates1}{notion}
\knowledge{\objBiStates2}{notion}
\newrobustcmd\hypQT[1][\setQ,\setT]{\kl[\hypQT]{\mathcal{H}(#1)}}
\knowledge\hypQT{notion}

\newrobustcmd\Epi{\ensuremath{\mathcal{E}}}
\newrobustcmd\Mono{\ensuremath{\mathcal{M}}}

\newrobustcmd\autoInit{\kl[\autoInit]{\mathcal{A}_{\mathit{init}}}}
\knowledge\autoInit{notion}

\newmathcommand\langLacc{\kl[\langLacc]{\mathcal{L}?}}
\knowledge\langLacc{notion}

\newmathcommand\varEpsAcc{\kl[\varEpsAcc]{\varepsilon?}}
\knowledge\varEpsAcc{notion}

\newrobustcmd\autoFinal{\kl[\autoFinal]{\mathcal{A}_{\mathit{final}}}}
\knowledge\autoFinal{notion}

\newrobustcmd\autA{\kl[\autA]{\mathcal{A}}}
\knowledge\autA{notion}

\newrobustcmd\autB{\kl[\autB]{\mathcal{B}}}
\knowledge\autB{notion}

\newrobustcmd\catAutoLKlT{\texorpdfstring{\Auto(\langL_{\KlTrans})}{Auto(LKl(T))}}
\newrobustcmd\catAutoLEMT{\Auto(\langLEMTTrans)}
\newrobustcmd\monadTrans{\kl[\monadTrans]{\mathcal{T}}}
\newrobustcmd\KlTrans{\texorpdfstring{\kl[\KlTrans]{\mathsf{Kl}(\mathcal{T})}}{Kl(T)}}
\knowledge\KlTrans{notion}
\newrobustcmd\KlT{\KlTrans}

\newrobustcmd\EMT{\kl[\EMT]{\mathsf{EM}(\mathcal{T})}}
\newrobustcmd\langLKlT{\texorpdfstring{\lang L_{\KlTrans}}{LKl(T)}}

\newrobustcmd\EpiKlT{\kl[\EpiKlT]{\mathcal{E}_{\mathsf{Kl}(\mathcal{T})}}}
\knowledge\EpiKlT{notion}
\newrobustcmd\MonoKlT{\kl[\MonoKlT]{\mathcal{M}_{\mathsf{Kl}(\mathcal{T})}}}
\knowledge\MonoKlT{notion}

\newrobustcmd\semTrans[1]{\kl[\semTrans]{[\![}#1\kl[\semTrans]{]\!]}}

\newrobustcmd\FreeTrans{\kl[\FreeTrans]{F_{\mathcal{T}}}}
\newrobustcmd\FreeTransEM{\kl[\FreeTransEM]{F^{\mathcal{T}}}}
\newrobustcmd\UTrans{\kl[\UTrans]{U_{\mathcal{T}}}}
\newrobustcmd\UTransEM{\kl[\UTransEM]{U^{\mathcal{T}}}}

\newrobustcmd\langLSetTrans{\kl[\langLSetTrans]{\lang L_{\Set}}}
\newrobustcmd\langLEMTTrans{\kl[\langLEMTTrans]{\lang L_{\EMT}}}

\newrobustcmd\pstar{\mathrel{\kl[\pstar]{\star}}}

\newrobustcmd\alphA{\kl[\alphA]{A}}
\knowledge\alphA{notion}

\newrobustcmd\alphB{\kl[\alphB]{B}}
\knowledge\alphB{notion}

\newrobustcmd\nto{\nrightarrow}
\newrobustcmd{\can}{\kl[\can]{\mathit{can}}}
\knowledge\can{notion}
\newrobustcmd{\id}{\kl[\id]{\mathit{id}}}
\knowledge\id{notion}
\newrobustcmd\Irr{\kl[\Irr]{\mathsf{Irr}}}
\newrobustcmd\IrrAB{\kl[\Irr]{\mathsf{Irr}}(A^*,B^*)}
\knowledge\Irr{notion}

\newrobustcmd\red{\kl[\red]{\mathsf{red}}}
\knowledge\red{notion}

\newrobustcmd\lcp{\kl[\lcp]{\mathsf{lcp}}}
\newrobustcmd\redL{\red(L)}
\knowledge\lcp{notion}

\newrobustcmd\lcpbis{\kl[\lcpbis]{\mathsf{lcp}}}
\knowledge\lcpbis{notion}
\newrobustcmd\redbis{\kl[\redbis]{\mathsf{red}}}
\knowledge\redbis{notion}

\newrobustcmd\epirightarrow{\twoheadrightarrow}%
\newrobustcmd\monorightarrow{\rightarrowtail}

\newrobustcmd\proj[1]{\kl[\proj#1]{\pi_{#1}}}
\knowledge{\proj1}{notion}
\knowledge{\proj2}{notion}

\newrobustcmd\wmin{w_{\mathit{min}}}
\knowledge{notion}
  | $\Lstar$-algorithm
  | $\Lstar$-like algorithm
  | $\Lstar$-like algorithms

\knowledge{notion}
 | hypothesis automaton@intro
 | hypothesis automata@intro

\knowledge{notion}
 | teacher@intro
 | teachers@intro
 | Teachers@intro

\knowledge{notion}
 | membership query@intro
 | membership queries@intro
 | Membership query@intro

\knowledge{notion}
 | equivalence query@intro
 | equivalence queries@intro
 | Equivalence query@intro

\knowledge{notion}
 | $T$-equivalences@intro
 | $T$-equivalence relations@intro
 | $T$-equivalence@intro
 | $T$-equivalent@intro
 | $T$-equivalence relations@intro

\knowledge{notion}
 | counter-example word@intro
 | counter-example@intro
 | counter-examples@intro
 | counter-example words@intro

\knowledge{notion}
 | closedness property@intro
 | closed@intro
 
\knowledge{notion}
 | consistency property@intro
 | consistent@intro

\knowledge{notion}
 | teacher

\knowledge{notion}
  | running instantiation
  | running instantiations
  | running examples
  | Running instantiation
  
\knowledge{notion}
  | input category
  | input categories

\knowledge{notion}
  | output category
  | output categories
  | computation category

\knowledge{notion}
  | word automaton
  | word automata

\knowledge{notion}
  | accept
  | accepts
  | accepted
  | accepting

\knowledge{notion}
  | minimal automaton
  | minimal object
  | minimal
  | $(\Epi,\Mono)$-minimal

\knowledge{notion}
 | $\Epi $-quotients
 | $\Epi $-quotient

\knowledge{notion}
 | $\Mono $-subobjects
 | $\Mono $-subobject

\knowledge{notion}
 | transduction
 | transductions
 | Transductions
 
\knowledge{notion}
 | irreducible
 | irreducible transduction
 | irreducible transductions
 | Irreducible transductions
 
\knowledge{notion}
 | subsequential transducer
 | subsequential transducers
 | Subsequential transducers
 | $(\KlT ,1,1)$-language
 
\knowledge{notion}
 | field weighted automata
 | Weighted automata over a field
 | weighted automata over the field $\K $
 | weighted automaton over $\K $
 | weighted automaton over the field $\K $

\knowledge{notion}
  | divide
  | divides
  | $(\Epi ,\Mono )$-divides
 
\knowledge{notion}
  | deterministic automaton
  | deterministic automata
  | DFA

\knowledge{notion}
  | $\langL$-automatable
  | automatable
  
\knowledge{notion} % C-automata
  | $\catC$-automaton@1
  | $\catC$-automata@1
  | automaton@1
  | automata@1
  | $\catC$-automaton
  | $\catC$-automata
  | automaton
  | automata

\knowledge{notion}
  | $\catC$-automata morphism
  | $\catC$-automata morphisms

\knowledge{notion} % (C,X,Y)-automata
  | $(\catC,X,Y)$-automaton
  | $(\catC,X,Y)$-automata
  | $(\catSet ,1,2)$-automaton
  | $(\catSet ,1,2)$-automata
  | $(\KlT,1,1)$-automaton
  | $(\KlT,1,1)$-automata
  | $(\catVec ,\K ,\K )$-automaton
  | $(\catVec ,\K ,\K )$-automata

\knowledge{notion} % C-languages
  | $\catC$-language@1
  | $\catC$-languages@1
  | language@1
  | languages@1
  | $\catC$-language
  | $\catC$-languages
  | language
  | languages

\knowledge{notion} % (C,X,Y)-languages
  | $(\catSet,1,2)$-language
  | $(\catC,X,Y)$-language
  | $(\catC,X,Y)$-languages

\knowledge{notion}
  | target language

\knowledge{notion}
  | evaluation query
  | Evaluation query
  | evaluation queries
  | Evaluation queries

\knowledge{notion}
  | equivalence query
  | Equivalence query
  | equivalence queries
  | Equivalence queries

\knowledge{notion}
  | counterexample
  | counterexamples

\knowledge{notion}
  | hypothesis automaton
  | hypothesis automata

\knowledge{notion}
  | coherence diagrams

\knowledge{notion}
  | consistent
  | consistency
  
\knowledge{notion}
  | $(\catC,X,Y)_{\setQ,\setT}$-biautomaton
  | ${(\setQ ,\setT )}$-biautomaton
  |$_{\setQ,\setT}$-biautomaton
  |$_{\setQ,\setT}$-biautomata
  | biautomaton
  | biautomata
  
\knowledge{notion}
  |$\catC_{\setQ,\setT}$-biautomaton
  |$\catC_{\setQ,\setT}$-biautomata
  |$\catTargC_{\setQ,\setT}$-biautomaton

\knowledge{notion}
  |$(\Epi,\Mono)$-realizable
 
\knowledge{notion}
  |minimal biautomaton
  |minimal biautomata

\knowledge{notion}
 | noetherianity
 | EMnoetherianity
 | $(\mathrm {Surjections},\mathrm {Injections})$-noetherian
 | $(\mathrm {Surjective\ linear\ maps},\mathrm {Injective\ linear\ maps})$-noetherian
 | $(\EpiKlT ,\MonoKlT )$-noetherian
 | $(\Epi ,\Mono )$-noetherianity
 | $(\Epi ,\Mono )$-noetherian

\knowledge{notion}
 | undefined

\knowledge{notion}
 | initial automaton

\knowledge{notion}
 | final automaton
 | final automata
 | Final automata

 \title{Learning automata and transducers:\\a categorical approach}
 \titlerunning{Learning automata and transducers: a categorical
   approach} \author{Thomas Colcombet}{\textsc{Irif} /
   \textsc{Cnrs}}{thomas.colcombet@irif.fr}{https://orcid.org/0000-0001-6529-6963}
 {Supported by the European Research Council (ERC) under the European
   Union’s Horizon 2020 research and innovation programme (grant
   agreement No.670624) and the DeLTA ANR project (ANR-16-CE40-0007)}
 \author{Daniela Petrişan}{\textsc{Irif} / Université de
   Paris}{daniela.petrisan@irif.fr}{https://orcid.org/0000-0001-9712-930X}{}
\author{Riccardo Stabile}{Università degli Studi di Milano, Dipartimento di Matematica}{riccardo.stabile@yahoo.com}{}{Supported by the European Commission under the Erasmus+ programme for a five-month study period at Université de Paris}
\authorrunning{T. Colcombet, D. Petrişan and R. Stabile}
\Copyright{Thomas Colcombet, Daniela Petrişan and Riccardo Stabile}



\nolinenumbers

\begin{document}

\maketitle

\begin{abstract} 
  In this paper, we present a categorical approach to learning
  automata over words, in the sense of the $\Lstar$-algorithm of
  Angluin. This yields a new generic $\Lstar$-like algorithm which can be
  instantiated for learning deterministic automata, automata weighted
  over fields, as well as subsequential transducers. The generic
  nature of our algorithm is obtained by adopting an approach in which
  automata are simply functors from a particular category representing
  words to a ``computation category''. We establish that the sufficient
  properties for yielding the existence of minimal automata (that were
  disclosed in a previous paper), in combination with some additional
  hypotheses relative to termination, ensure the correctness of our
  generic algorithm.
\end{abstract}

\section{Introduction}
\label{sec:intro}

Learning automata is a classical subject at the intersection of
machine learning and automata theory. It has found a wide range of
applications spanning from adaptive model checking, compositional
verification to learning network invariants or interface
specifications for Java classes. We refer the reader
to~\cite{LearningMeetsVerification} and the references therein for a
survey of such applications.

\AP The most famous learning algorithm for automata is certainly the
""$\Lstar$-algorithm"" of Angluin \cite{angluin87:learning}. Its goal is
to learn a regular language of words~$L$. For this, the algorithm
interacts with a ""teacher@@intro"" (an oracle) who knows~$L$ by
asking two kinds of queries: 
 
\vspace{-0.7em}
\begin{description}
  \itemAP[""Membership query@@intro""] it can ask whether a given word
  belongs to~$L$, or \itemAP[""Equivalence query@@intro""] it can
  provide a ""hypothesis automaton@@intro"" and ask the
  "teacher@@intro" whether this automaton recognizes~$L$ or not. If
  the answer is no, the "teacher@@intro" is bound to provide a
  ""counter-example word@@intro"", witnessing the non-equivalence.
\end{description}
The algorithm stops when the "teacher@@intro" agrees that the
"hypothesis automaton@@intro" recognizes the language~$L$. A key
property of the "$\Lstar$-algorithm" is that it terminates in time
polynomial in the size of the alphabet, of the minimal deterministic
automaton for~$L$, and of the longest "counter-example@@intro". Furthermore,
all candidate automata appearing during its execution (and hence in
particular the final one) are deterministic, complete, and minimal.
\vspace{-1em}
\paragraph*{The "$\Lstar$-algorithm"}
Let us illustrate the behaviour of this algorithm when it tries to
learn the language~$\{a\}$ over the alphabet~$\Sigma=\{a\}$.
\AP At each
step, the algorithm maintains two sets of words~$\intro*\setQintro,\intro*\setTintro$, starting
with~$\setQintro=\{\varepsilon\}$, $\setTintro=\{\varepsilon\}$. One can understand~$\setQintro$
as a set of words which identify states of the "hypothesis
automaton@@intro" under construction. The set~$\setTintro$ is used in order to
discover if words need to be distinguished by the automaton: two
words~$u,v\in \Sigma^*$ are ""$T$-equivalent@@intro"" if for
all~$t\in \setTintro$, $ut\in L$ if and only if~$vt\in L$.

\AP At the beginning, the algorithm attempts to construct an automaton
with as sole state $\varepsilon\in \setQintro$, which has to be initial.  In
particular, the target of the transition labelled~$a$ issued
from~$\varepsilon$ has to be determined. Such a transition should go
to a state in~$\setQintro$ which has to be "$T$-equivalent@@intro"
to~$\varepsilon a=a$. \AP It fails since there are no such states
in~$\setQintro$ (we say that the pair $\setQintro,\setTintro$ fails to have the ""closedness
property@@intro""). The algorithm corrects it by adding the word~$a$
to~$\setQintro$.  We reach~$\setQintro=\{\varepsilon,a\}$, $\setTintro=\{\varepsilon\}$. The
algorithm now tries to construct an automaton with 
states~$\setQintro=\{\varepsilon,a\}$: this time, it is possible to construct
an $a$-labelled transition from~$\varepsilon\in \setQintro$ to~$a\in \setQintro$. What
should now be the $a$-labelled transition issued from~$a$? It should
be some state~$q\in \setQintro$ which is "$T$-equivalent@@intro"
to~$aa$. Luckily, there is one, namely~$\varepsilon$. Hence, we
succeed in constructing the left "hypothesis automaton@@intro" in
\cref{figure:candidates}.
\begin{wrapfigure}{r}{8.5cm}
	\vspace{-0.35cm}
	\begin{tikzpicture}[node distance=2cm,on grid,auto,baseline,initial text={}]
	  \node[state,initial above]  (0)                 {$\varepsilon$};
	  \node[state,accepting by arrow,accepting above]                    (1) [right of=0]  {$a$};
	  \path[->]
	  (0) edge     [bend left]           node {$a$}  (1)
	  (1) edge  	[bend left]		     node {$a$}  (0);
	\end{tikzpicture}\qquad%
		\begin{tikzpicture}[node distance=2cm,on grid,auto,baseline,initial text={}]
	  \node[state,initial above]  (0)                 {$\varepsilon$};
	  \node[state,accepting by arrow,accepting above]                    (1) [right of=0]  {$a$};
	  \node[state]                    (2) [right of=1]  {$aa$};
	  \path[->]
	  (0) edge                node {$a$}  (1)
	  (1) edge  			 node {$a$}  (2)
	  (2) edge [loop above] node {$a$} ();
	\end{tikzpicture}
	\caption{Two successive "hypothesis automata@@intro"}
	\label{figure:candidates}
\end{wrapfigure}
The algorithm now "queries for equivalence@equivalence query" of the language of this automaton with the language.
This fails since~$a(aa)^*\neq L=\{a\}$, and
hence the "teacher@@intro" answers in return a
"counter-example word@@intro", say~$aaa$. The algorithm then adds (for
reasons that are not detailed here) the prefix~$aa$ of~$aaa$ to~$\setQintro$,
yielding $\setQintro=\{\varepsilon,a,aa\}$, $\setTintro=\{\varepsilon\}$. Here,
$\varepsilon$ and~$aa$ are "$T$-equivalent@@intro", but constructing
an $a$-labeled transition from~$\varepsilon$ would yield~$a$, while
constructing one from~$aa$ would yield~$aaa$, which is
"$T$-equivalent@@intro" to~$\varepsilon$.  \AP Hence, $\varepsilon$
and $aa$ cannot be merged as a same state (we say that $\setQintro,\setTintro$ fails to
have the ""consistency property@@intro""). The algorithm compensates
it by adding~$a$ to~$\setTintro$, thus yielding~$\setQintro=\{\varepsilon,a,aa\}$
and~$\setTintro=\{\varepsilon,a\}$. Now, $\setQintro,\setTintro$ are both "closed@@intro" and
"consistent@@intro", and the right "hypothesis automaton@@intro" in
\cref{figure:candidates} is constructed. It recognizes~$\{a\}$, and
thus the "teacher@@intro" agrees and the algorithm terminates. It has
constructed the minimal deterministic and complete automaton for the
language~$L=\{a\}$.

This example witnesses the different steps involved in the algorithm:
(a) if $(\setQintro,\setTintro)$ is not "closed@@intro", a word is added to~$\setQintro$, (b)
if~$(\setQintro,\setTintro)$ is not "consistent@@intro", a word is added to~$\setTintro$, and (c)
when $(\setQintro,\setTintro)$ is both "closed@@intro" and "consistent@@intro", it is
possible to construct a "hypothesis automaton@@intro" and perform an
"equivalence query@@intro": if this automaton happens to not accept
the expected language, the "teacher@@intro" provides a
"counter-example word@@intro" from which words to add to~$\setQintro$ are
constructed. The algorithm functions by performing the operation until
the "teacher@@intro" agrees.

The correctness of the algorithm bears many resemblances with the
question of minimizing deterministic automata. This can be witnessed
in the fact that the "$\Lstar$-algorithm" automatically constructs
minimal deterministic and complete automata. It can also be witnessed
in the fact that the "$T$-equivalences@@intro" induce along the run finer and
finer partitions of the words that converge eventually to the
Myhill-Nerode equivalence, another notion highly connected to
minimization questions.

The "$\Lstar$-algorithm" turns out to be extremely robust, and has been
extended to various other forms of automata (weighted automata over
fields~\cite{DBLP:conf/ciac/BergadanoV94,DBLP:journals/siamcomp/BergadanoV96},
nominal automata \cite{nominal-learning17}, omega automata
\cite{omega-learning16}, non-deterministic automata
\cite{NFA-learning09}, alternating automata
\cite{alternating-learning}, symbolic automata
\cite{symbolic-learning17}, subsequential transducers
\cite{Vilar96,Vilar00}, transducers of trees
\cite{rational-learning12,DTtree-transducers-learning16}).
Although with a focus on concrete implementations, Bollig et
al.~\cite{libalf} emphasize that ``the need for a unifying framework
collecting various types of learning techniques is, thus, beyond all
questions.''
\vspace{-1em}
\paragraph*{Contributions}

The aim of this paper is to present such a unifying framework for
learning word automata using the toolkit of category theory.
Concretely, we provide an abstract categorical version of Angluin's
"$\Lstar$-algorithm", called $\FunL$ (\cref{algorithm:main}), we prove
its correctness and termination (\cref{thm:algo-corr-term}), and we
give three "running instantiations" for it, namely in the case of
deterministic automata, "field weighted automata" and "subsequential
transducers".

\AP To this end, we reuse the framework developed
in~\cite{colcopetr20:automin} which models automata as functors from
an "input category" $\intro*\catIintro$ (describing the structure of the computation) to an
"output category" $\intro*\catCintro$. 
\AP For example, to model word automata, the
"input category" is a fixed three-object category $\catI$, that we will
recall in Section~\ref{sec:min}. By varying the category $\catCintro$, this definition captures several forms of automata, and in particular
the ones mentioned above. In \cite{colcopetr20:automin}, we present
sufficient conditions on $\catCintro$ that guarantee the existence of
minimal automata. These conditions are quite mild: $\catCintro$ should
have certain products and coproducts, on one hand, and a factorization
system, on the other. Apart from these three conditions on the "output category", \cref{thm:algo-corr-term} -- which states that our new
algorithm computes the minimal automaton for the language to be
learned -- requires only one additional assumption which ensures termination, namely a
`finiteness' hypothesis (using the notion of "noetherianity").

In order to describe our generic $\FunL$ algorithm we provide abstract
versions of the steps of the "$\Lstar$-algorithm" described above. These
are obtained as follows:
\begin{itemize}
\item We describe the pair of sets of words $(\setQintro,\setTintro)$ using a
  four-object category $\catIQT$, introduced in
  Definition~\ref{def:catIQT}. This category is a modification of
  $\catI$, which allows us to obtain a partial view of the
  language, namely only its values on words of the form $qt$ and $qat$ with
  $q\in \setQintro, t\in \setTintro$ and $a$ a letter in the alphabet.
\item Computing the approximations of the Myhill-Nerode equivalence
  (that is, the "$T$-equivalence relations@@intro") roughly corresponds in
  our generic setting to performing a minimization-like
  computation. This is achieved using off-the-shelf results
  from~\cite{colcopetr20:automin} by changing the input category from
  $\catI$ to $\catIQT$. We obtain a form of minimal `biautomaton'
  featuring an $\varepsilon$-transition between its two state objects.
\item The pair $(\setQintro,\setTintro)$ being "closed@@intro" and "consistent" amounts to the
  $\varepsilon$-transition of the above biautomaton being an
  isomorphism between the two state objects. We then say that the pair
  $(\setQintro,\setTintro)$ is "$\langL$-automatable". Under this assumption, it is
  meaningful to collapse the two state objects, defining in this way
  the "hypothesis automaton" (represented now as a functor
  $\catI\to \catCintro$).
\end{itemize}
What is interesting about our $\FunL$-algorithm -- compared to
previous approaches -- is that it highlights the strong relationship
between learning and minimizing automata. Each elementary step of the
algorithm involves performing a minimization-like computation and
leverages the modularity of our previous
work~\cite{colcopetr20:automin}, this time by varying the input
category. In contrast to other category theoretic approaches to
learning, \FunL\  does not rely neither on algebras nor on
coalgebras. Instead, we exploit the symmetry of the word automata
model. This is reflected by the self-duality of the "input category" $\catI$, which is underpinning the well known duality between observability and reachability.

Finally, a prominent instantiation of the \FunL-algorithm is Vilar's
learning algorithm of subsequential transducers~\cite{Vilar96}. Our
notion of "$\langL$-automatable" pairs $(Q,T)$ perfectly instantiates
to the conditions considered by Vilar to construct a hypothesis
transducer. A coalgebraic modelisation of subsequential transducers
was provided in~\cite{DBLP:journals/iandc/Hansen10}, but, to the best
of our knowledge, this example is not featured in the category
theoretic learning literature.
\vspace{-1em}
\paragraph*{Related works}

We briefly review the (co)algebraic approaches to automata learning
that have been proposed in recent years. The
paper~\cite{DBLP:conf/birthday/JacobsS14a} was the first to recast key
ingredients of Angluin's algorithm in a coalgebraic setting.
This line of work was continued with the CALF framework of van Heerdt
et.al~\cite{CALF17}, which models automata  as triples
consisting of an algebra for a functor, an initial map and an output
map. In~\cite[Section~5]{CALF17} a connection between minimization and
learning is mentioned and formalized for DFAs. More precisely, the
main theorem proving the correctness of the learning algorithm
~\cite[Theorem~16]{CALF17} can be used to show the correctness of the
minimization algorithm for DFA, with reachability and observability
playing a crucial role. The same authors proposed in~\cite{CALF20} a
learning algorithm for automata with side-effects. These are
extensions of DFAs to automata interpreted in a category of
Eilenberg-Moore algebras for a $\catSet$ monad $T$ -- used to represent a
certain side effect. For example, the finite powerset monad
corresponds to non-determinism and the ensuing automata model serves
to represent non-deterministic automata. In order to prove the
termination of the learning algorithm, the monad $T$ above is assumed
to preserve finite sets. Hence the monad that we use in the present
work to model subsequential transducers does not fit in the scope
of~\cite{CALF20}. Another small difference is that we work within the
Kleisli category.

Another category theoretic learning algorithm was proposed
in~\cite{DBLP:conf/fossacs/BarloccoKR19} and provides a coalgebraic
and duality theoretic foundation for learning bisimilarity quotients
of state-based transition systems. The core idea is to use logical
formulas as tests, taking stock of dual adjunctions between states and
logical theories, formalized as algebra-coalgebra dualities.

The recent paper~\cite{urbat2019automata} gives a learning algorithm
for automata whose transitions can be encoded both as algebras for a
functor $F$ on a category $\mathcal{C}$, and as coalgebras for the
right adjoint of $F$ (assumed to exist).  The approximations of the
Myhill-Nerode equivalence present in the learning algorithm are
computed using factorizations of morphisms from approximations of an
initial algebra (obtained using an initial chain) to approximations of
a final coalgebra (obtained using a final co-chain). Some of the
ingredients of this category theoretic algorithm are similar to ours,
e.g. the heavy use of factorization systems or the notion of
``finite'' object in a category, however, there are more assumptions on
the underlying category and on the preservation properties of the
adjoint functors considered in the (co)-algebraic definition of the
automata, see~\cite[Assumption~3.5]{urbat2019automata}
and~\cite[Assumption~4.1]{urbat2019automata}.
\vspace{-2em}
\paragraph*{Structure of the paper}
In~\cref{sec:min}, we present necessary material from
\cite{colcopetr20:automin}, which includes in particular the
categorical modeling of automata, its instantiation for deterministic
finite automata, for "field weighted automata" and for "subsequential
transducers", and how to minimize them.  This material is key in our
description of the algorithm in \cref{sec:core-alg}. For simplicity,
we describe first a slightly simplified version of the algorithm. The
optimized version is the subject of
\cref{sec:opt-algorithm}. \Cref{sec:conclusion} concludes.

\section{Minimization}
\label{sec:min}

In this section, we recall the categorical approach to automata
minimization from~\cite{colcopetr17:automin,colcopetr20:automin}.

\subsection{Languages and automata as functors}
We first recall the notion of automata as functors.  \AP We consider
an arbitrary small category $\intro*\catIstart$, called the ""input
category"", and one of its full subcategories $\intro*\catOstart$,
denoting by $i$ the inclusion functor:
\begin{tikzcd}
  \catOstart \arrow[hookrightarrow]{r}{i} & \catIstart.
\end{tikzcd}
Intuitively, $\catIstart$ represents the the inner computations
performed by an "automaton@@1", and in particular its internal
behaviour, while $\catOstart$ represents the observable behaviour of
the "automaton@@1" and is used to define the "language@@1" it
"accepts".

\AP We consider another category $\intro*\catC$, called the ""output
category"", which models the output computed by the automaton (e.g., a
boolean value, probabilities, words over an alphabet).
\begin{definition}
  \label{def:cat-auto}
  \AP A ""$\catC$-automaton@@1"" (or simply an "automaton@@1")
  $\intro*\autA$ is a functor from $\catIstart$ to $\catC$.  A
  ""$\catC$-language@@1"" (or simply a "language@@1") $\intro*\langL$
  is a functor from $\catOstart$ to $\catC$. A "$\catC$-automaton@@1"
  $\autA$ ""accepts"" a "$\catC$-language@@1" $\langL$ if
  $\autA \circ i=\langL$.

  We denote by $\intro*\catAuto(\langL)$ the subcategory of the
  functor category $[\catIstart,\catC]$:
  \begin{itemize}
  \item whose objects are all "$\catC$-automata@@1" $\autA$
    "accepting" $\langL$; \itemAP whose arrows are ""$\catC$-automata
    morphisms"", meaning natural transformations
    $\alpha \colon \autA_1 \Rightarrow \autA_2$ such that
    $\alpha \circ \textit{id}_i = \textit{id}_{\langL}$.
  \end{itemize}
\end{definition}

\AP In this paper, we will instantiate the "input category"
$\catIstart$ in two ways. The first one, $\catI$, is used
in~\cite{colcopetr17:automin} to model different forms of ""word
automata""; we describe it in this section and use it for modeling the
three "running instantiations".  In Section~\ref{sec:core-alg}, we
will consider another "input category", $\catIQT$, which we use in the
process of constructing our "hypothesis automata".

\begin{wrapfigure}[4]{r}{0.3\linewidth}
  \centering
  \vspace{-1em}
  \begin{tikzcd}[baseline=1cm]
    \intro*\objIn \arrow[rightarrow]{r}{\intro*\symbLeft} &
    \intro*\objStates \arrow[loop]{u}[swap]{\intro*\symbA}
    \arrow[rightarrow]{r}{\intro*\symbRight} & \intro*\objOut.
  \end{tikzcd}
\end{wrapfigure}
\AP We define now the "input category" $\intro*\catI$ used for
describing "word automata". Here $\intro*\alphA$ is a finite alphabet, 
fixed for the rest of the paper, and $\alphA^*$ the set of words over
it. The "input category" $\catI$ is the category freely generated by
the graph on the right, where $\symbA$ ranges over $\alphA$:
That is, $\catI$ is the three-object category with arrows spanned by
$\symbLeft$, $\symbRight$ and $\symbA$ for all $a \in \alphA$, so that
the composition \begin{tikzcd} \objStates \arrow[rightarrow]{r}{w} &
  \objStates \arrow[rightarrow]{r}{w'} & \objStates \end{tikzcd} is
given by the concatenation $ww'$.  So, for example, the morphisms in
$\catI$ from $\objIn$ to $\objStates$ are of the form $\symbLeft w$
with $w\in\alphA^*$, while the morphisms on the object $\objStates$
are of the form $w$ with $w\in\alphA^*$.

\AP Let $\intro*\catO$ denote the full subcategory of $\catI$ on the
objects $\objIn$ and $\objOut$. Its morphisms are of the
form \begin{tikzcd} \objIn \arrow[rightarrow]{r}{\symbLeft
    w\symbRight} & \objOut
\end{tikzcd} for~$w\in\alphA^*$.

\AP Hereafter, by a "language@@1" we mean a functor from $\catO$ to
$\catC$ and by an "automaton@@1" we mean a functor from $\catI$ to
$\catC$. If $\langL(\objIn)=X$ and $\langL(\objOut)=Y$, a "language"
$\langL$ will be referred to as a ""$(\catC,X,Y)$-language""; if
$\autA(\objIn)=X$ and $\autA(\objOut)=Y$, an "automaton" $\autA$ will
be called a ""$(\catC,X,Y)$-automaton"".
\AP We provide three ""running instantiations"" of the "output
category" $\catC$ and of the objects $X$ and $Y$, in order to model
deterministic automata, "field weighted automata" and "subsequential
transducers".

\begin{example}[Deterministic automata]
  A deterministic and complete automaton is a
  "$(\catSet,1,2)$-automaton". Indeed, we can see a functor
  $\autA\colon\catI\to\catSet$ with $\autA(\objIn)=1$ and
  $\autA(\objOut)=2$ as a deterministic automaton by interpreting
  \begin{itemize}
  \item $\autA(\objStates)$ as its set of states,
  \item $\autA(\symbLeft)\colon 1\to \autA(\objStates)$ as choosing
    the initial state,
  \item $\autA(\symbA)\colon \autA(\objStates)\to \autA(\objStates)$
    as the transition map for the letter $\symbA\in\alphA$,
  \item $\autA(\symbRight)\colon \autA(\objStates)\to 2$ as the
    characteristic map of the subset of accepting states.
  \end{itemize}
\end{example}

\begin{example}[Weighted automata over a field] Let $\intro*\K$ be a
  field and let $\catVec$ denote the corresponding category of
  $\K$-vector spaces and linear transformations.  \AP A ""weighted
  automaton over the field $\K$"" (in the sense of
  \cite{schutzenberger61}) is a "$(\catVec,\K,\K)$-automaton". Indeed,
  a functor $\autA\colon\catI\to\catVec$ with $\autA(\objIn)=\K$ and
  $\autA(\objOut)=\K$ is seen as a "weighted automaton over $\K$" by
  interpreting
  \begin{itemize}
  \item $\autA(\objStates)$ as the vector space spanned by its states,
  \item $\autA(\symbLeft)\colon \K\to \autA(\objStates)$ as the linear
    transformation mapping the unit of $\K$ to the initial vector,
  \item $\autA(\symbA)\colon \autA(\objStates)\to \autA(\objStates)$
    as the linear transformation transition for the letter
    $\symbA\in\alphA$,
  \item $\autA(\symbRight)\colon \autA(\objStates)\to \K$ as the
    output linear transformation.
  \end{itemize}
\end{example}

\begin{example}["Subsequential transducers"] \AP The aim is to
  represent what we call ""transductions"" in this paper, which are
  partial maps from~$\alphA^*$ to~$\alphB^*$, where~$\intro*\alphB$ is
  some fixed output alphabet.  \AP Roughly, a "subsequential
  transducer" \cite{Choffrut79} is a deterministic automaton which, at
  each step, while reading an input letter from~$\alphA$, either has
  no transition or has a unique transition which changes
  deterministically the state and outputs a word from~$\alphB^*$. In
  this paper, we define ""subsequential transducers"" as
  "$(\KlT,1,1)$-automata" \cite{colcopetr20:automin}, for a definition
  of~$\KlT$ that we give now.
  
  The \textbf{"output category" $\KlT$.}  Let $\intro*\monT$ be the
  monad defined by $\monT X=\alphB^*\times X+1$ and let $\intro*\KlT$
  denote the Kleisli category for $\monT$.  Concretely, the objects of
  $\KlT$ are sets, while its morphisms, denoted by negated arrows, are
  of the form $f\colon X\nrightarrow Y$ for a function
  $f\colon X\to\monT Y$, that is, a partial function from $X$ to
  $\alphB^*\times Y$.  \AP We write $\bot$ for the element of the
  singleton $1$ and think of it as the ""undefined"" element. Given
  $f\colon X\nrightarrow Y$ and $g\colon Y\nrightarrow Z$, their
  composite $g\circ f\colon X\nrightarrow Z$ is defined on $x\in X$ by
  $(uv,z)$, when $f(x)=(u,y)\in \alphB^*\times Y$ and
  $g(y)=(v,z)\in \alphB^*\times Z$ (with $uv$ denoting the
  concatenation of $u$ and $v$ in $\alphB^*$) and $f(x)=\bot$ in all
  other cases. Note that with this definition, a "transduction" can be
  identified in an obvious manner with a map from~$\alphA^*$ to arrows
  of the form $1\nrightarrow 1$.

  \AP We now recall that "$(\KlT,1,1)$-automata" are equivalent to
  \reintro{subsequential transducers} 
  in the sense of Choffrut's definition \cite{Choffrut79}.  Indeed, we
  can see a functor $\autA\colon\catI\to\KlT$ with $\autA(\objIn)=1$
  and $\autA(\objOut)=1$ as a subsequential transducer by interpreting
  \begin{itemize}
  \item $\autA(\objStates)$ as the set of states,
  \item $\autA(\symbLeft)\colon 1\nto \autA(\objStates)$ as either
    choosing an initial state together with an initial output in
    $\alphB^*$ or having an undefined initial state,
  \item $\autA(\symbA)\colon \autA(\objStates)\nto \autA(\objStates)$
    as the transition map for the letter $\symbA$ which associates to
    a given state either "undefined" or a pair consisting of an output
    word in $\alphB^*$ and a successor state,
  \item $\autA(\symbRight)\colon \autA(\objStates)\nto 1$ as the final
    map which associates to a state either its output in $\alphB^*$ or
    "undefined" when it is non-accepting.
  \end{itemize}
\end{example}

\subsection{Minimization of automata}

Now we describe what it means to be "minimal" in a category
(\cref{definition:division}) together with an abstract result of
existence of such an object (\cref{lemma:minimality}).  We then
provide sufficient material for our three "running instantiations" to
be covered.

\AP Let $\intro*\catK$ be a category endowed with a factorization
system $(\Epi,\Mono)$. We write
\begin{tikzcd}
  {}\ar[r,two heads]& {}
\end{tikzcd}
for arrows belonging to $\Epi$ and we will call them
""$\Epi$-quotients""; we write
\begin{tikzcd}
  {}\ar[r,rightarrowtail]& {}
\end{tikzcd}
for arrows belonging to $\Mono$ and we will call them
""$\Mono$-subobjects"".

\begin{definition}\label{definition:division}
  Consider two objects $X,Y$ of $\mathcal{K}$. We say that $X$
  ""$(\Epi,\Mono)$-divides"" $Y$ whenever $X$ is an "$\Epi$-quotient"
  of an "$\Mono$-subobject" of $Y$, that is, we have a span of the
  form:
  \begin{center}
    \begin{tikzcd}
      X & \cdot \arrow[l, two heads] \arrow[r, tail] & Y\,.

    \end{tikzcd}

  \end{center}

  \AP An object $Z$ in $\mathcal{K}$ is ""$(\Epi,\Mono)$-minimal"" if
  it "$(\Epi,\Mono)$-divides" all the objects in $\mathcal{K}$.
  
\end{definition}

As shown in the following lemma, having an initial and a final object
turns out to be a sufficient condition for the minimal object to exist
and be unique up to
isomorphism. 

\begin{lemma}\label{lemma:minimality}
  \label{minobj} \AP Let $\catK$ be a category endowed with an initial
  object $I$, a final object $F$ and a factorization system
  $(\Epi,\Mono)$. Let $\intro*\objMinStart$ be the factorization of
  the unique arrow from $I$ to $F$:
  \begin{center}
    \begin{tikzcd}
      I \arrow[twoheadrightarrow]{r} & \objMinStart
      \arrow[rightarrowtail]{r} & F.
    \end{tikzcd}
  \end{center}
  Then $\objMinStart$ is "$(\Epi,\Mono)$-minimal".
\end{lemma}

We apply this lemma when $\catK$ is instantiated with a category of
automata $\catAuto(\langL)$.
 
\begin{corollary}
  \AP
  \label{cor:min}
  If the category $\catAuto(\langL)$ has an initial "automaton@@1"
  $\intro*\autoInit(\langL)$, a final "automaton@@1"
  $\intro*\autoFinal(\langL)$ and a factorization system, then the
  "minimal automaton" $\intro*\objMin(\langL)$ for the "language@@1"
  $\langL$ is obtained via the following factorization:
  \begin{tikzcd}
    \autoInit(\langL) \arrow[twoheadrightarrow]{r} & \objMin(\langL)
    \arrow[rightarrowtail]{r} & \autoFinal(\langL).
  \end{tikzcd}
\end{corollary}

Notice that this notion of minimization is parametric in the
factorization system. 
In all our examples, we obtain a suitable factorization system on
$\catAuto(\langL)$ from one on $\catC$, as
follows.

\begin{lemma}
  \label{lem:fact-cat-auto}
  If a category $\catC$ has a factorization system $(\Epi, \Mono)$,
  then the category $\catAuto(\langL)$ has a factorization system
  $(\Epi_{\catAuto(\langL)},\Mono_{\catAuto(\langL)})$, where
  $\Epi_{\catAuto(\langL)}$ consists of all natural transformations
  with components in $\Epi$ and $\Mono_{\catAuto(\langL)}$ consists of
  all natural transformations with components in $\Mono$.
\end{lemma}

\begin{example}[factorization systems]%
  \label{example:KlT-factorization}
  There exists a factorization system in our three "running
  examples". For $\catSet$, this is the well known factorization
  system $(\textrm{Surjections},\textrm{Injections})$.  Similarly in
  $\catVec$,
  $(\mathrm{Surjective\ linear\ maps},\mathrm{Injective\ linear\
    maps})$ is a factorization
  system.
  
  In the case of $\KlT$, the factorization system does not follow from
  general arguments.  We define now the factorization system
  $(\EpiKlT,\MonoKlT)$ for~$\KlT$.  \AP Given a morphism
  $f\colon X\nrightarrow Y$ in $\KlTrans$, we write
  $\intro*\proj1(f)\colon X\to B^*+\{\bot\}$ and
  $\intro*\proj2(f)\colon X\to Y+\{\bot\}$ for the projections:
  if~$f(x)=\bot$ then $\proj1(x)=\proj2(x)=\bot$, otherwise
  $f(x)=(\proj1(f)(x),\proj2(f)(x))$.

  \AP The class $\intro*\EpiKlT$ consists of all the morphisms of the
  form $e\colon X\nrightarrow Y$ such that $\proj2(e)$ is surjective
  (i.e. for every $y \in Y$ there exists $x \in X$ so that
  $\proj2(e)(x)=y$) and the class $\intro*\MonoKlT$ consists of all
  the morphisms of the form $m\colon X\nrightarrow Y$ such that
  $\proj2(m)$ is injective and $\proj1(m)$ is the constant function
  mapping every $x\in X$ to $\varepsilon$.

  By~\cite[Lemma~4.8]{colcopetr20:automin}, $(\EpiKlT,\MonoKlT)$ is a
  factorization
  system.
\end{example}

We specialize the result of Corollary~\ref{cor:min} to the case of
word automata $\catI\to\catC$. Due to the special shape of the
category $\catI$, we can compute the initial and the final automata,
provided the "output category" satisfies some mild assumptions,
recalled in Lemmas~\ref{lem:init-aut} and~\ref{lem:final-aut}.

\begin{lemma}
  \label{lem:init-aut}
  \AP Fix a "language" $\langL\colon\catO\to\catC$.  If the category
  $\catC$ has countable copowers of $\langL(\objIn)$, the ""initial
  automaton"" $\autoInit(\langL)$ exists and is given by the following
  data:
  \begin{itemize}
  \item
    $\displaystyle\autoInit(\langL)(\objStates)=\coprod_{\mathclap{\alphA^*}}\langL(\objIn)$;
  \item
    $\displaystyle\autoInit(\langL)(\symbLeft)\colon\langL(\objIn)\to\coprod_{\mathclap{\alphA^*}}\langL(\objIn)$
    is given by the coproduct injection corresponding to
    $\varepsilon$, for this reason we will denote this map by
    $\varepsilon$;
  \item
    $\displaystyle\autoInit(\langL)(\symbA)\colon\coprod_{\mathclap{\alphA^*}}\langL(\objIn)\to\coprod_{\mathclap{\alphA^*}}\langL(\objIn)$
    is given on the $w$-component $\langL(\objIn)$ by the coproduct
    injection corresponding to $w\symbA$;
  \item
    $\displaystyle\autoInit(\langL)(\symbRight)\colon\coprod_{\mathclap{\alphA^*}}\langL(\objIn)\to\langL(\objOut)$
    is the coproduct of the morphisms
    $\langL(\symbLeft w\symbRight)\colon
    \langL(\objIn)\to\langL(\objOut)$ with $w\in A^*$, that is, it
    computes the value of the language on a given word, for this
    reason we will also denote this map by $\intro*\langLacc$.
  \end{itemize}
\end{lemma}

\begin{example}
  \label{example:initial-automata}
  Since the categories $\catSet$, $\catVec$ and $\KlT$ have all
  copowers, the "initial automaton" for a given "language" can be
  easily computed in these cases as an instance of the above lemma. We
  recall the details for $\catSet$ and $\KlT$.

  \centering{ \begin{tikzcd}[baseline=(current bounding
      box.north),outer sep=0pt,inner sep=0pt,column sep=5em] 1 \arrow[r,
      "\varepsilon"] & \alphA^* \arrow["w\mapsto wa"', loop,
      distance=2em, in=125, out=55] \arrow[r, "\langLacc"] & 2
    \end{tikzcd} \qquad\qquad\qquad
    \begin{tikzcd}[column sep=5em]
      1 \arrow[r,negated, "({\varepsilon,\varepsilon})"] & \alphA^*
      \arrow["w\mapsto ({\varepsilon,wa})"', loop, negated,
      distance=2em, in=125, out=55] \arrow[r, negated, "\langLacc"] &
      1
    \end{tikzcd}
  }

  \begin{itemize}
  \item Given a language $\langL\colon\catO\to\catSet$, the initial
    deterministic automaton accepting $\langL$ is described above in
    the left diagram.  Its state space is the set of all words, with
    $\varepsilon$ being the initial one. A word is accepted if and
    only if it belongs to the language.
  \item For a language $\langL\colon\catO\to\KlT$, the initial
    "subsequential transducer" accepting $\langL$ is as depicted in
    the right diagram.  Its state space is $\alphA^*$, the initial
    state is $\varepsilon\in\alphA^*$ with initial output
    $\varepsilon\in\alphB^*$. For an input letter $\symbA\in\alphA$,
    the corresponding transition maps $w$ to $wa$ and produces output
    $\varepsilon\in\alphB^*$.  Finally, the map $\langLacc$, which is
    in fact a function from $\alphA^*$ to $\alphB^*+1$, associates to
    a word $w$ the value of the language at $w$, that is, the value
    computed by $\langL(\symbLeft w\symbRight)$.
  \end{itemize}
\end{example}

\begin{lemma}
  \label{lem:final-aut}
  \AP Fix a "language" $\langL\colon\catO\to\catC$.  If the category
  $\catC$ has countable powers of $\langL(\objOut)$, the ""final
  automaton"" $\autoFinal(\langL)$ exists and is given by the
  following data:
  \begin{itemize}
  \item
    $\displaystyle\autoFinal(\langL)(\objStates)=\prod_{\mathclap{\alphA^*}}\langL(\objOut)$;
  \item
    $\displaystyle\autoFinal(\langL)(\symbLeft)\colon\langL(\objIn)\to\prod_{\mathclap{\alphA^*}}\langL(\objOut)$
    is the product of the morphisms
    $\langL(\symbLeft w\symbRight)\colon
    \langL(\objIn)\to\langL(\objOut)$ with $w\in A^*$, for this reason
    we will also denote this map by $\langL$;
  \item
    $\displaystyle\autoFinal(\langL)(\symbA)\colon\prod_{\mathclap{\alphA^*}}\langL(\objOut)\to\prod_{\mathclap{\alphA^*}}\langL(\objOut)$
    is the product over $w\in\alphA^*$ of the $\symbA w$-projections
    $\displaystyle\prod_{\mathclap{\alphA^*}}\langL(\objOut)\to\langL(\objOut)$;
  \item
    $\displaystyle\autoFinal(\langL)(\symbRight)\colon\coprod_{\mathclap{\alphA^*}}\langL(\objOut)\to\langL(\objOut)$
    is given by the $\varepsilon$-projection, for this reason we will
    also denote this map by $\intro*\varEpsAcc$.
  \end{itemize}
\end{lemma}

\begin{example}[the "final automata" in $\catSet$, $\catVec$ and
  $\KlT$]
  \label{example:final-automata}
  Since the categories $\catSet$ and $\catVec$ have all products, the
  "final automaton" for a given "language" can be computed using
  \cref{lem:final-aut}. We illustrate this for $\catSet$ and $\KlT$.

  \centering{\begin{tikzcd}[column sep=5em] 1 \arrow[r, "\langL"] & 2^{\alphA^*}
      \arrow["K\mapsto a^{-1}K"', loop, distance=2em, in=125, out=55]
      \arrow[r, "K\mapsto K(\varepsilon)"] & 2
    \end{tikzcd} \qquad\qquad\qquad
    \begin{tikzcd}[column sep=5em]
      1 \arrow[r,negated, "({\lcp(\langL),\red(\langL)})"] & \IrrAB
      \arrow["{K\mapsto(\lcp(K),\red(K))}"', loop, negated,
      distance=2em, in=125, out=55] \arrow[r, negated, "K\mapsto
      K(\varepsilon)"] & 1
    \end{tikzcd}
  }
      
  \begin{itemize}
  \item Given a language $\langL\colon\catO\to\catSet$, the final
    deterministic automaton accepting $\langL$ is described above in
    the left diagram.  Its state space is the set of all languages
    over the alphabet~$\alphA$. The initial state is the
    language~$\langL$ itself. A language is an accepting state if and
    only if it contains~$\varepsilon$.  Given a language~$K$, while
    reading letter~$a$, the automaton goes to the residual
    language~$a^{-1}K = \{u\in K\mid au\in K\}$.
  \item Somewhat suprisingly, $\KlT$-automata also fit in the scope of
    \cref{lem:final-aut}, as we can prove that the object $\IrrAB$
    (which we will define next) is the power of $\alphA^*$-many copies
    of $1$ in $\KlT$.
    \AP Define first, given a "transduction"~$K$, $\intro*\lcp(K)$ to
    be undefined if~$K$ is nowhere defined, and the longest common
    prefix of the words in $\{K(u)\mid u\in\alphA^*\}$ otherwise.  \AP
    A "transduction"~$K$ is ""irreducible"" if~$\lcp(K)=\varepsilon$.
    We denote by ~$\intro*\IrrAB$ the set of "irreducible
    transductions".  \AP For all~$K$ not nowhere defined, we put
    $\intro*\red(K)$ to be the only "irreducible transduction" such
    that~$K(u)=\lcp(K)\red(K)(u)$, i.e. the "transduction" in which
    the prefix~$\lcp(K)$ has been stripped away from all
    outputs. For~$K$ nowhere defined, let $\reintro*\red(K)$ be also
    nowhere defined.

    We can describe now the "final automaton" for a
    "transduction"~$\langL$ as an automaton that has "irreducible
    transductions" as states.  The initial map is the constant map
    equal to~$({\lcp(\langL),\red(\langL)})$ (or undefined if~$\langL$
    is nowhere defined).  
    When reading the letter~$a$ from state~$K$, the
    automaton jumps to~$\red(K(a-))$ in which~$K(a-)$ is such
    that~$K(a-)(u)=K(au)$ (or undefined if~$K(a-)$ is nowhere
    defined). The final map sends an "irreducible
    transduction" to $K(\varepsilon)$.
  \end{itemize}
\end{example}

\begin{wrapfigure}{r}{5.3cm}\vspace{-.5cm}
  \begin{tikzcd}[row sep = 2.8em, column sep=1.4em]
    &\displaystyle\coprod_{\mathclap{\alphA^*}}\langL(\objIn) \arrow[twoheadrightarrow]{d}{\intro*\epiMin} \arrow[rightarrow,bend left]{rd}{\langLacc}& \\
    \langL(\objIn) \arrow[rightarrow, bend right]{rd}[swap]{\langL}
    \arrow[rightarrow, bend left]{ru}{\varepsilon}
    \arrow[rightarrow]{r} &
    \objMin(\langL)(\objStates)
    \arrow[rightarrow]{r} \arrow[rightarrowtail]{d}{\intro*\monoMin} & \langL(\objOut) \\
    &\displaystyle\prod_{\mathclap{\alphA^*}}\langL(\objOut)
    \arrow[rightarrow, bend right]{ru}[swap]{\varEpsAcc}&
  \end{tikzcd}\vspace{-1.6cm}
\end{wrapfigure}
Combining Lemmas~\ref{lem:fact-cat-auto},~\ref{lem:init-aut}
and~\ref{lem:final-aut} with Corollary~\ref{cor:min}, we obtain the
following result.
\begin{theorem}
  \label{minimalwordauto}
  Let $\catC$ be a category with a factorization system
  $(\Epi, \Mono)$ and let $\langL\colon\catO\to\catC$ be a
  "language". Suppose $\catC$ has all countable copowers of
  $\langL(\objIn)$ and all countable powers of $\langL(\objOut)$.  The
  "minimal" "$\catC$-automaton@@1" $\objMin(\langL)$ "accepting"
  $\langL$ is obtained via the factorization in the commuting diagram
  to the right.
\end{theorem}

\section{The basic \FunL{} algorithm}
\label{sec:core-alg}

In this section, we provide our generic \FunL{-}algorithm for learning
"word automata". Just as in "Angluin's algorithm@$\Lstar$-algorithm",
there are a teacher and
a learner. Throughout this section we fix the alphabet $\alphA$, the
output category $\catC$ and its factorization systems $(\Epi,\Mono)$,
all known to both teacher and learner.
\AP The ""teacher"" knows a "language" $\langL\colon\catO\to\catC$,
hereafter called the ""target language"". The learner wants to find
this language, the output of the algorithm being the "minimal
automaton" $\objMin(\langL)$ "accepting"
$\langL$. The learner can ask two kinds of queries, which can be thought as
high-level generalizations of Angluin's original ones in the special
case of deterministic automata (see \cite{angluin87:learning}).
\begin{itemize}
\item \AP ""Evaluation queries"": given a certain word $w$, what is
  $\langL(\symbLeft w \symbRight)$?
\item \AP ""Equivalence queries"": does a certain "automaton" "accept"
  the "target language"? If it does not, what is a "counterexample"
  for it not doing that?
\end{itemize}

\AP Let $\mathcal{A}$ be an "automaton" which is incorrect, that is,
such that $\mathcal{A}\circ i \neq \langL$, $\langL$ being the "target
language"; a word $w$ is said to be a ""counterexample"" if
$\mathcal{A}\circ i (\symbLeft w \symbRight) \neq \langL(\symbLeft w
\symbRight)$.
In other words, a "counterexample" witnesses the incorrectness of a
certain "automaton" proposed by the learner.

In order to formulate the generic algorithm, we still need to
generalize the notions of table and "hypothesis automaton" from
Angluin's original algorithm. We do this in
Section~\ref{sec:hypothesis-automata}. We provide the generic
algorithm and prove its correctness and termination in
Section~\ref{subsection:learning-algorithm}.
\subsection{Hypothesis automata}
\label{sec:hypothesis-automata}

Just as in Angluin's "$\Lstar$-algorithm", the learner keeps in memory
a pair $(\intro*\setQ,\intro*\setT)$ of subsets of $\alphA^*$ such
that $\setQ$ is prefix-closed, i.e. it contains the prefixes of all
its elements, while $\setT$ is suffix-closed, the same for the
suffixes; in particular, $\varepsilon \in \setQ \cap \setT$.
Using the "evaluation queries", the learner produces an approximation
of $\objMin(\langL)$, explicitly a "hypothesis automaton", to be
introduced in
Definition~\ref{def:hypothesis-automaton}.

It turns out that the category $\catAuto(\langL)$ does not suffice to
capture the whole learning process. At a given stage of the algorithm,
the learner has access, via "evaluation queries", only to a part
of $\langL$: specifically, he knows the values of
$\langL(\symbLeft qt\symbRight)$ and
$\langL(\symbLeft qat \symbRight)$, where $q\in Q$, $t\in T$ and
$a\in\alphA$. This leads us to consider a restriction of the language
$\langL$ to a subcategory of $\catO$ whose arrows are of the form
$\symbLeft qt\symbRight$ or $\symbLeft qat\symbRight$ as above. To
produce a "hypothesis automaton" "consistent" with this partial view of
$\langL$, we would also need to adapt the "input category". A first
attempt would be to discard some of the arrows of $\catI$ from
$\objIn$ to $\objStates$, respectively from $\objStates$ to
$\objOut$. Explicitly, we would like to keep only the arrows of the
form $\symbLeft q\colon\objIn\to\objStates$ for the state words
$q\in Q$ and, respectively, $t \symbRight\colon\objStates\to\objOut$ for
the test words $t\in T$. However, this is not feasible: we would also
need the transition maps $\symbA\colon\objStates\to\objStates$, and via composition we would generate, for example, all arrows
$\symbLeft w\colon\objIn\to\objStates$. The solution is to
``dissociate'' the state object $\objStates$ in $\catI$ and
consider a four-state input category.

\begin{definition}
  \label{def:catIQT}
  \AP The "input category" $\intro*\catIQT$ is the free category
  generated by the graph
  \begin{center}
    \begin{tikzcd}
      \objIn \arrow[rightarrow]{r}{\intro*\symbLeftQ} &
      \intro*\objBiStates1 \arrow[rightarrow, shift
      left]{r}{\intro*\symbAQT} \arrow[rightarrow, shift
      right]{r}[swap]{\intro*\symbEpsilon} & \intro*\objBiStates2
      \arrow[rightarrow]{r}{\intro*\symbRightT} & \objOut
    \end{tikzcd}
  \end{center}
  for all $q \in \setQ$, $a \in \alphA$, $t \in \setT$ and with
  $\varepsilon$ a fixed symbol (informally representing the empty
  word) such that the following \AP ""coherence diagrams"" commute for
  all $a\in\alphA$, for all $q\in Q$ such that $qa\in Q$, and for all
  $t\in T$ such that $at\in T$:
  \begin{center}
    \begin{tikzcd}[row sep=0em]
      & \objBiStates1 \arrow[rightarrow]{rd}{\symbAQT} & &     & \objBiStates2 \arrow[rightarrow]{rd}{\symbRightT} & \\
      \objIn \arrow[rightarrow]{ru}{\symbLeftQ} \arrow[rightarrow]{rd}[swap]{\symbLeftQ\symbAQT} & & \makebox[1.5em][r]{$\objBiStates2$};  &     \makebox[1.5em][l]{$\objBiStates1$} \arrow[rightarrow]{ru}{\symbAQT} \arrow[rightarrow]{rd}[swap]{\symbEpsilon} & & \objOut.  \\
      & \objBiStates1 \arrow[rightarrow]{ru}[swap]{\symbEpsilon} & & &
      \objBiStates2 \arrow[rightarrow]{ru}[swap]{\symbAQT\symbRightT}
      &
    \end{tikzcd}
  \end{center}
  \AP Furthermore, let $\intro*\catOQT$ denote the full subcategory of
  $\catIQT$ on the objects $\objIn$ and $\objOut$.
\end{definition}

\AP The two "coherence diagrams" in the definition of $\catIQT$, as well
as the prefix-closure of $Q$ and the suffix-closure of $T$, ensure
that we have a functor \[\intro*\iStar\colon\catIQT\to\catI\] which merges
$\objBiStates1$ and $\objBiStates2$ sending both of them to
$\objStates$, maps $\varepsilon\colon \objBiStates1\to \objBiStates2$
to the identity on $\objStates$ and maps all the other morphisms of
$\catIQT$ to the homonymous ones in $\catI$.

\begin{lemma}
  \label{lem:input-cats}
  The functor $\iStar\colon\catIQT\to\catI$ is well defined and,
  furthermore, $\catOQT$ is a subcategory of $\catO$. That is, we have
  the following commuting diagram:
  \begin{center}
    \begin{tikzcd}
      \catOQT \arrow[d, hook] \arrow[r, hook] & \catO \arrow[d, hook] \\
      \catIQT \arrow[r, "\iStar"] & \catI.
    \end{tikzcd}
  \end{center}
\end{lemma}

\AP The partial knowledge of the language $\langL$
 the learner
has access to at this given stage of the algorithm is captured by the
restriction $\intro*\langLQT$ of $\langL$ to $\catOQT$:
\begin{center}
  \begin{tikzcd}
    \langLQT\colon\catOQT\ar[r,hook]&\catO \ar[r,"\langL"]&\catC.
  \end{tikzcd}
\end{center}

Hence, to a pair $(Q,T)$ we can associate the category
$\catAuto(\langLQT)$ obtained by instantiating in
Definition~\ref{def:cat-auto} the "input category" $\catIstart$ with
$\catIQT$ and its observable behaviour subcategory with
\begin{tikzcd}
  \catOQT\ar[r,hook]&\catIQT\,.
\end{tikzcd}

\begin{definition}
  \AP We call a functor $\intro*\autB$ in $\catAuto(\langLQT)$ a
  ""${(\setQ,\setT)}$-biautomaton"" $\autB$ or a
  ""$\catC_{\setQ,\setT}$-biautomaton"", if we want to underline the
  dependence on $\catC$. We say that $\autB$ is ""consistent""
  with the "$\catC$-language" $\langL$.
\end{definition}

In the "$\Lstar$-algorithm", the learner constructs a table associated to
each pair of subsets $(Q,T)$. This is done essentially by computing
the quotient of the state words in $Q$ by an approximation
$\sim_{\setT}$ of the Myhill-Nerode equivalence for a language $L$
given by: $w\sim_{\setT} v$ iff for all $t\in T$ we have
$wt\in L\Leftrightarrow vt\in L$. \AP This leads us to consider as a
generalization of the notion of \emph{table}
 the
""minimal biautomaton"" $\objMin(\langLQT)$ in the category
$\catAuto(\langLQT)$. In order to compute it, we use Corollary~\ref{cor:min}. To this end, we first exhibit
explicitly the initial and the final objects of $\catAuto(\langLQT)$,
assuming that the "output category" $\catC$ has got certain products
and coproducts.

We will use the following notation.  Given two subsets $R$ and $S$ of
$\alphA^*$, let $RS$ denote the set
$\left \{ xy \colon x \in R, y \in S \right \}$.

\begin{lemma}
  \label{lem:initial-QT-biautomaton}
  \AP Assume $\catC$ has all countable copowers of $\langL(\objIn)$. The initial
  "$\catC_{\setQ,\setT}$-biautomaton" is the functor
  $\autoInit(\langLQT)\colon \catIQT\to\catC$ described in the next
  diagram
  \begin{center}
    \begin{tikzcd}[column sep=3em]
      \langL(\objIn) \arrow[rightarrow]{r}{\intro*\symbLeftQInit} &
      \displaystyle\coprod_{\mathclap{\setQ}}\langL(\objIn)
      \arrow[rightarrow, shift left]{r}{\intro*\symbAQTInit}
      \arrow[rightarrow, shift
      right]{r}[swap]{\intro*\symbEpsilonInit} &
      \displaystyle\coprod_{\mathclap{\setQ\cup
          \setQ\alphA}}\langL(\objIn)
      \arrow[rightarrow]{r}{\intro*\symbRightTInit} & \langL(\objOut),
    \end{tikzcd}
  \end{center}
  where, explicitly:
  \begin{itemize}
  \item
    $\displaystyle\autoInit(\langLQT)(\objBiStates1)=\coprod_{\setQ}\langL(\objIn)$
    and
    $\displaystyle\autoInit(\langLQT)(\objBiStates2)=\coprod_{\setQ\cup
      \setQ\alphA}\langL(\objIn)$;
  \item $\symbLeftQInit:=\autoInit(\langLQT)(\symbLeftQ)$ is the
    coproduct injection $j_q$ of $\langL(\objIn)$ into
    $\displaystyle\coprod_{\mathclap{\setQ}}\langL(\objIn)$;
  \item $\symbEpsilonInit:=\autoInit(\langLQT)(\symbEpsilon)$ is the
    canonical inclusion between the two coproducts;
  \item $\symbAQTInit:=\autoInit(\langLQT)(\symbA)$ is obtained via
    the universal property as the coproduct over $q\in \setQ$ of the
    canonical injections
    $\displaystyle
    j_{qa}\colon\langL(\objIn)\to\coprod_{\mathclap{\setQ\cup
        \setQ\alphA}}\langL(\objIn)$;
  \item $\symbRightTInit:=\autoInit(\langLQT)(\symbRightT)$ is
    obtained via the universal property as the coproduct over
    $w\in \setQ\cup \setQ\alphA$ of the morphims
    $\langL(\symbLeft w t \symbRight)\colon\langL(\objIn)\to
    \langL(\objOut)$.
  \end{itemize}
\end{lemma}

Dually, we can describe the final "$\catC_{\setQ,\setT}$-biautomaton" as follows.
\begin{lemma}
  \label{lem:final-QT-biautomaton}
  \AP Assume $\catC$ has all countable powers of
  $\langL(\objOut)$. The final
  "$\catC_{\setQ,\setT}$-biautomaton" is the functor
  $\autoFinal(\langLQT)\colon
  \catIQT\to\catC$ described in the next diagram
  \begin{center}
    \begin{tikzcd}[column sep=3em]
      \langL(\objIn) \arrow[rightarrow]{r}{\intro*\symbLeftQFinal} &
      \displaystyle\prod_{\mathclap{\setT \cup \alphA
          \setT}}\langL(\objOut) \arrow[rightarrow, shift
      left]{r}{\intro*\symbAQTFinal} \arrow[rightarrow, shift
      right]{r}[swap]{\intro*\symbEpsilonFinal} &
      \displaystyle\prod_{\mathclap{\setT}}\langL(\objOut)
      \arrow[rightarrow]{r}{\intro*\symbRightTFinal} &
      \langL(\objOut),
    \end{tikzcd}
  \end{center}
  where, explicitly:
  \begin{itemize}
  \item
    $\displaystyle\autoFinal(\langLQT)(\objBiStates1)=\prod_{\setT\cup
      \alphA\setT}\langL(\objOut)$ and
    $\displaystyle\autoInit(\langLQT)(\objBiStates2)=\coprod_{\setT}\langL(\objOut)$;
  \item
    $\symbLeftQFinal:=\autoFinal(\langLQT)(\symbLeftQ)$ is obtained
    via the universal property as the product over $w\in \setT\cup
    \alphA \setT$ of the morphisms $\langL(\symbLeft
    qw\symbRight)\colon\langL(\objIn)\to\langL(\objOut)$;
  \item
    $\symbEpsilonFinal:=\autoFinal(\langLQT)(\symbEpsilon)$ is the
    canonical restriction  
    between the two products;
  \item
    $\symbAQTFinal:=\autoFinal(\langLQT)(\symbA)$ is obtained via the
    universal property of
    $\displaystyle\prod_{\mathclap{\setT}}\langL(\objOut)$ as the
    product over $t\in
    \setT$ of the canonical projections
    $\displaystyle\pi_{at}\colon\prod_{\mathclap{\setT\cup \alphA
        \setT}}\langL(\objOut)\to\langL(\objOut)$;
  \item
    $\symbRightTFinal:=\autoFinal(\langLQT)(\symbRightT)$ is the
    projection
    $\displaystyle\pi_t\colon\prod_{\mathclap{\setT}}\langL(\objOut)\to\langL(\objOut)$.
  \end{itemize}
\end{lemma}

Combining Corollary~\ref{cor:min} with
Lemmas~\ref{lem:fact-cat-auto},~\ref{lem:initial-QT-biautomaton}
and~\ref{lem:final-QT-biautomaton}, we obtain the "minimal biautomaton"
$\objMin(\langLQT)$ in $\catAuto(\langLQT)$.

\begin{theorem}
  \label{thm:minimal-QT-biautomaton} Assume that $\catC$ is equipped
  with a factorization system $(\Epi,\Mono)$ and has countable copowers of
  $\langL(\objIn)$ and countable powers of $\langL(\objOut)$. Then the minimal
  "$\catC_{\setQ,\setT}$-biautomaton" $\objMin(\langLQT)$ is obtained
  as the unique up to isomorphism factorization of the unique morphism
  from $\autoInit(\langLQT)$ to $\autoFinal(\langLQT)$.
  \begin{center}
    \begin{tikzcd}[column sep=3em]
      & \displaystyle\coprod_{\mathclap{\setQ}}\langL(\objIn)
      \arrow[twoheadrightarrow]{d}{\intro*\epiOne} \arrow[rightarrow,
      shift left]{r}[yshift=0.2ex]{\symbAQTInit} \arrow[rightarrow,
      shift right]{r}[swap, yshift=-0.4ex]{\symbEpsilonInit} &
      \displaystyle\coprod_{\mathclap{\setQ\cup \setQ\alphA}}\langL(\objIn) \arrow[twoheadrightarrow]{d}{\intro*\epiTwo} \arrow[rightarrow, bend left]{rd}{\symbRightTInit} & \\
      \langL(\objIn) \arrow[rightarrow]{r}{\intro*\symbLeftQMin} \arrow[rightarrow, bend left]{ru}{\symbLeftQInit} \arrow[rightarrow, bend right]{rd}[swap]{\symbLeftQFinal}& \objMin(\langLQT)(\objBiStates1) \arrow[rightarrowtail]{d}{\intro*\monoOne} \arrow[rightarrow, shift left]{r}[yshift=0.2ex]{\intro*\symbAQTMin} \arrow[rightarrow, shift right]{r}[swap,yshift=-0.4ex]{\intro*\symbEpsilonMin} &\objMin(\langLQT)(\objBiStates2)\arrow[rightarrow]{r}{\intro*\symbRightTMin} \arrow[rightarrowtail]{d}{\intro*\monoTwo} &\langL(\objOut)\\
      & \displaystyle\prod_{\mathclap{\setT \cup \alphA
          \setT}}\langL(\objOut) \arrow[rightarrow, shift
      left]{r}[yshift=0.2ex]{\symbAQTFinal} \arrow[rightarrow, shift
      right]{r}[swap, yshift=-0.4ex]{\symbEpsilonFinal} &
      \displaystyle\prod_{\mathclap{\setT}}\langL(\objOut)
      \arrow[rightarrow, bend right]{ru}[swap]{\symbRightTFinal} &
    \end{tikzcd}
  \end{center}
\end{theorem}

Notice that the arrows $\symbLeftQMin$, $\symbAQTMin$,
$\symbEpsilonMin$ and $\symbRightTMin$ are obtained using
the diagonal fill-in property of the factorization system.  Let us now
see how this theorem instantiates in the case of deterministic
automata.
\begin{example}
  Assume the "target language" $\langL$ is a
  "$(\catSet,1,2)$-language", so the learner wants to learn the
  minimal deterministic automaton accepting $\langL$. For a given
  couple $(\setQ,\setT)$, the "minimal biautomaton" is obtained as the
  following factorization.
  
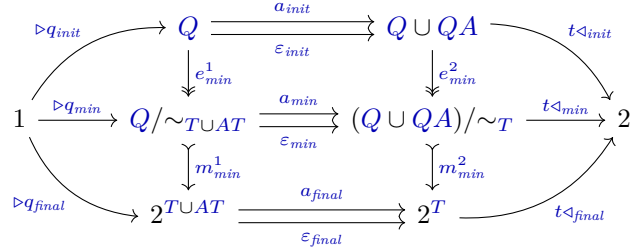
\begin{center}
    \begin{tikzcd}[column sep=3em]
      & \setQ \arrow[twoheadrightarrow]{d}{\epiOne}
      \arrow[rightarrow, shift left]{r}[yshift=0.2ex]{\symbAQTInit}
      \arrow[rightarrow, shift right]{r}[swap,
      yshift=-0.4ex]{\symbEpsilonInit} &
      \setQ\cup \setQ\alphA \arrow[twoheadrightarrow]{d}{\epiTwo} \arrow[rightarrow, bend left]{rd}{\symbRightTInit} & \\
      1 \arrow[rightarrow]{r}{\symbLeftQMin} \arrow[rightarrow,
      bend left]{ru}{\symbLeftQInit} \arrow[rightarrow, bend
      right]{rd}[swap]{\symbLeftQFinal}& \setQ/{\sim_{\setT \cup
          \alphA
          \setT}} \arrow[rightarrowtail]{d}{\monoOne} \arrow[rightarrow, shift left]{r}[yshift=0.2ex]{\symbAQTMin} \arrow[rightarrow, shift right]{r}[swap,yshift=-0.4ex]{\symbEpsilonMin} &(\setQ\cup \setQ\alphA)/{\sim_{\setT}}\arrow[rightarrow]{r}{\symbRightTMin} \arrow[rightarrowtail]{d}{\monoTwo} & 2\\
      & 2^{\setT \cup \alphA \setT} \arrow[rightarrow, shift
      left]{r}[yshift=0.2ex]{\symbAQTFinal} \arrow[rightarrow, shift
      right]{r}[swap, yshift=-0.4ex]{\symbEpsilonFinal} & 2^{\setT}
      \arrow[rightarrow, bend right]{ru}[swap]{\symbRightTFinal} &
    \end{tikzcd}
  \end{center}
  Hence, the first set of states of the "minimal biautomaton" is the set
  $\setQ$ quotiented by the $\setT \cup \alphA \setT$-approximation
  $\sim_{\setT \cup \alphA \setT}$ of the Myhill-Nerode
  equivalence. The second set of states is the quotient of
  $\setQ\cup \setQ\alphA$ by $\sim_{\setT}$. These kinds of quotients are
  also needed in the classical Angluin's algorithm, when building the
  table corresponding to the couple $(\setQ,\setT)$.

  Let us now understand when the map $\symbEpsilonMin$ is an
  isomorphism, that is, in this case, a bijection. We can verify that
  $\symbEpsilonMin$ being a surjection is equivalent to the table in
  "$\Lstar$-algorithm" being closed, that is, for all $q\in\setQ$ and
  $a\in\alphA$ there exists $q'\in \setQ$ such that $q'\sim_{\setT} qa$. On
  the other hand, $\symbEpsilonMin$ being an injection is equivalent
  to the consistency of the table from Angluin's "$\Lstar$-algorithm". It
  means that if $q$ and $q'$ are such that $q\sim_{\setT} q'$ then
  $q\sim_{\setT \cup \alphA \setT} q'$. 
\end{example}

If the table from Angluin's algorithm is generalized via the "minimal
biautomaton" $\objMin(\langLQT)$, the above example suggests that the
conditions that make possible the generation of a hypothesis
automaton from a table can be stated at this abstract level by
requiring the morphism $\symbEpsilonMin$ be an isomorphism.
 In this
way, we can identify $\objMin(\langL)(\objBiStates1)$ and
$\objMin(\langL)(\objBiStates2)$ to obtain the state space of the
"hypothesis automaton".

\begin{definition}
  \label{def:hypothesis-automaton}
  \AP If the map $\symbEpsilonMin$ is an isomorphism, we say that
  $(\setQ,\setT)$ is ""$\langL$-automatable"".  The ""hypothesis
  automaton"" $\intro*\hypQT$ associated to a "$\langL$-automatable"
  couple $(\setQ,\setT)$ is the "$\catC$-automaton" with state space
  $\objMin(\langLQT)(\objBiStates1)$ described on the generator arrows
  of $\catI$ by % \marginpar{not yet fully happy with this formulation}
  
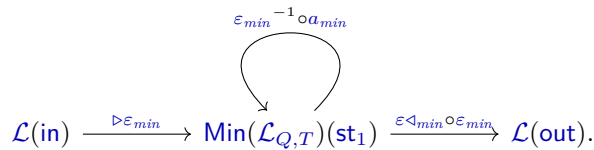
\begin{center}
    \begin{tikzcd}[column sep=4em]
      \langL(\objIn) \arrow[rightarrow]{r}{\symbLeftQMin[\varepsilon]}
      & \objMin(\langLQT)(\objBiStates1)
      \arrow[loop]{u}[swap]{\symbEpsilonMin^{-1} \circ \symbAQTMin}
      \arrow[rightarrow]{r}{\symbRightTMin[\varepsilon]\circ\symbEpsilonMin}
      & \langL(\objOut).
    \end{tikzcd}
  \end{center}
\end{definition}
  
The uniqueness up to isomorphism of the "hypothesis automaton"
$\hypQT$ is an easy consequence of the uniqueness up to isomorphism
of the "minimal biautomaton" in $\catAuto(\langLQT)$.
  
It is important to remark that, when passing from a "biautomaton" to
an "automaton", the "consistency" with the "language" is preserved, in
the sense of the lemma below.
\begin{lemma}
  \label{consistencypreserved} Let $(\setQ,\setT)$ be an
  "$\langL$-automatable" couple and let $\hypQT$ be its associated
  "hypothesis automaton". Then the next diagram commutes:
  \begin{center}
        \begin{tikzcd}
          \catOQT \arrow[r, hook] \arrow[rr, bend left, "\langLQT"] & \catI \arrow[r, "\hypQT"'] &\catC. 
        \end{tikzcd}

  \end{center}
\end{lemma}

\subsection{The learning algorithm}
\label{subsection:learning-algorithm}
We now have all the necessary ingredients to state the
\intro*\FunL{-}algorithm.  Our algorithm takes as input a "target language"
$\langL$ and outputs its "minimal automaton", provided some mild
assumptions listed in Theorem~\ref{thm:algo-corr-term} are satisfied.
We start by instantianting the couple $(\setQ,\setT)$ by
$(\varepsilon,\varepsilon)$. As long as this couple is not
"$\langL$-automatable", further words are added to the subsets $\setQ$
and $\setT$ to force $\symbEpsilonMin$ to become an isomorphism. Once
this is achieved, we obtain a "hypothesis automaton". If this
automaton does not recognize the "target language", then the provided
"counterexample" and its prefixes are added to $\setQ$, in order to
let the learner progress in learning.

While the role played by "equivalence queries" is self-evident, notice
that "evaluation queries" are necessary in order to build up the
category $\catAuto(\langLQT)$ and analyse its "minimal automaton".

  \begin{algorithm}
  \caption{The basic \FunL learning algorithm}
  \label{algorithm:main}
  \SetKwInOut{Input}{input}\SetKwInOut{Output}{output}
  \Input{minimally adequate teacher of the "target language" $\langL$}
  \Output{$\objMin(\langL)$}
  $\setQ:=\setT:=\left \{ {\varepsilon} \right \}$
	
  \Repeat{the answer is yes}{ \While{$(\setQ,\setT)$ is not
      "$\langL$-automatable" \label{line3}}{ \If{$\symbEpsilonMin \notin
        \Epi$ \label{line4}}{add $\setQ\alphA$ to $\setQ$ \label{line5}}
      \If{$\symbEpsilonMin \notin \Mono$ \label{line7}}{add
        $\alphA \setT$ to $\setT$ \label{line8}}} ask an "equivalence query" for the
    "hypothesis automaton" $\hypQT$
		
    \If{the answer is no \label{line12}}{add the provided
      "counterexample" and all its prefixes to $\setQ$ \label{line13}}
  }
  \Return{$\hypQT$}
\end{algorithm}

In order for this generic algorithm to work, we need several mild
assumptions on the "output category" $\catC$ and on the "target
language". First, in order to compute the "hypothesis automaton" we
need the existence of the "minimal automaton" in the category
$\catAuto(\langLQT)$. For this reason, we will assume the hypothesis
of Theorem~\ref{thm:minimal-QT-biautomaton}, pertaining to the
existence of certain powers, certain copowers and a factorization
system. Furthermore, in order to ensure the termination of our
algorithm, a "noetherianity" condition is required on the language
$\langL$, akin to the regularity of the language in the
$\Lstar$-algorithm.  This notion, also used
in~\cite{urbat2019automata}, can be understood as a finiteness
assumption as shown in \cref{example:finite-is-noetherian}.

\begin{definition}\label{definition:EM-noetherian}
  \AP An object $X$ of $\catC$ is called ""$(\Epi,\Mono)$-noetherian""
  when the following conditions hold.
  \begin{itemize}
  \item There does not exist an infinite co-chain of "$\Epi$-quotients"
    of $X$ as in the left commutative diagram below and such that the
    arrows $e_1,e_2\ldots \in \Epi$ are not
    isomorphisms.
  \item There does not exist an infinite chain of "$\Mono$-subobjects"
    of $X$ as in the right commutative diagram below and such that the
    arrows $m_1,m_2\ldots \in \Mono$ are not isomorphisms.
    
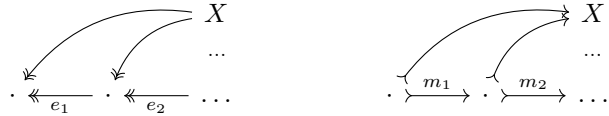
\begin{center}
      \begin{tikzcd}
      &                                   & X \arrow[lld, two heads, bend right] \arrow[ld, two heads, bend right] \arrow[two heads, bend right,color=white]{d}[color=black]{\ldots} &  &                                                           &                                                          & X                      \\
\cdot & \cdot \arrow[l, "e_1", two heads] & \ldots \arrow[l, "e_2", two heads]                                                                      &  & \cdot \arrow[r, "m_1", tail] \arrow[rru, tail, bend left] & \cdot \arrow[r, "m_2", tail] \arrow[ru, tail, bend left] & \ldots \arrow[bend right,color=white, tail]{u}[color=black]{\ldots}
\end{tikzcd}
    \end{center}
  \end{itemize}
\end{definition}

\begin{example}\label{example:finite-is-noetherian}
   Let us see now what "noetherianity@EMnoetherianity" means for the factorization systems of our "running instantiations" (\cref{example:KlT-factorization}).
   It is easy to see that in $\catSet$, an object~$X$ is "$(\mathrm{Surjections},\mathrm{Injections})$-noetherian" if and only if it is finite in the usual sense. Similarly, an object of~$\catVec$ is "$(\mathrm{Surjective\ linear\ maps},\mathrm{Injective\ linear\ maps})$-noetherian" if and only if it is a finite dimension vector space.
   With a bit more thoughts, one can establish that an object~$X$ of~$\KlT$ is "$(\EpiKlT,\MonoKlT)$-noetherian" if and only if it is finite.
\end{example}

In order to guarantee the termination of our algorithm, we require the
"$(\Epi,\Mono)$-noetherianity" of the state space of the "minimal
automaton" of the target language. This is a natural condition,
generalizing the regularity of the target language in the
"$\Lstar$-algorithm". If "$(\Epi,\Mono)$-noetherian" objects are
closed under "$\Epi$-quotients" and "$\Mono$-subobjects" -- as it is the
case in all our examples -- we could also replace this hypothesis by
assuming the existence of an automaton with
"$(\Epi,\Mono)$-noetherian" state space which accepts the target
language.

\begin{theorem}
  \label{thm:algo-corr-term}
  We consider a target language $\langL\colon\catO\to\catC$ such
  that:
  \begin{itemize}
  \item the "output category" $\catC$ is endowed with a
    factorization system $(\Epi, \Mono)$;
  \item $\catC$ has all copowers of $\langL(\objIn)$ and all
    powers of $\langL(\objOut)$;
  \item the state space $\objMin(\langL)(\objStates)$ of the "minimal
    automaton" for $\langL$ is "$(\Epi,\Mono)$-noetherian".
  \end{itemize}
  Then the \FunL{-}algorithm terminates, eventually producing
  the "minimal automaton" $\objMin(\langL)$ "accepting" the "target
  language".
\end{theorem}

The proof of this theorem relies on a careful analysis of the
factorizations
\begin{center}
\begin{tikzcd}
  \coprod_{\setQ}\langL(\objIn) \arrow[r, two heads] & \factQT
  \arrow[r, tail] & \prod_{\setT}\langL(\objOut)
\end{tikzcd}
\end{center}
of the canonical maps
$\coprod_{\setQ}\langL(\objIn)\to\prod_{\setT}\langL(\objOut)$
obtained by taking the coproduct over $q\in Q$ of the product over
$t\in T$ of $\langL(\symbLeft q t\symbRight)$. We can prove that the
state spaces of the "biautomata" featured while running the algorithm
are precisely of the form $\factQT$, while the state space of the
"minimal automaton" accepting $\langL$ is $\fact{\alphA^*,\alphA^*}$. 

We prove that the \textbf{while} loop terminates in
Proposition~\ref{prop:term-while-loop}. In Lemma~\ref{algotermination}
we show that only finitely many counterexamples can be added, hence
the algorithm terminates. Finally, the fact that the outcome automaton
is minimal is shown in Lemma~\ref{lem:correctness-algo}.

Next, we see how the \FunL-algorithm instantiates to the case of
"subsequential transducers". We need to understand what it means for
$\symbEpsilonMin$ to be an isomorphism.

\begin{example}[Learning algorithm for subsequential transducers]
  Assume the "target language" $\langL$ is a "$(\KlT,1,1)$-language",
  so the learner wants to learn the minimal "subsequential transducer"
  accepting $\langL$. \AP We need to extend the notions of $\lcp$ and
  $\red$ to a generic partial map $g$ whose domain is
  $\setT\subseteq\alphA^*$ and whose codomain is $\alphB^*$ as
  follows: $\intro*\lcpbis(g)$ is undefined if $g$ is nowhere defined,
  and denotes the longest common prefix of the words in
  $\{g(u)\mid u\in\setT\}$ otherwise; analogously,
  $\intro*\redbis(g)\colon\setT\rightarrow\alphB^* \cup \{ \bot \}$ is
  nowhere defined if $g$ is nowhere defined, and is the only partial
  map such that~$g(u)=\lcpbis(g)\redbis(g)(u)$ otherwise. Thinking of
  the "language" to learn as a "transduction"
  $f\colon \alphA^*\to \alphB^*+\{\bot\}$, let's define the following
  equivalence relation for all $q_1,q_2 \in \setQ$:
  $ q_1 \sim_{\setT} q_2$ if and only if
  $\redbis(f(q_1-)\arrowvert_{\setT})(t)=\redbis(f(q_2-)\arrowvert_{\setT})(t)$
  for all $t \in \setT$, $f(q-)\arrowvert_{\setT}$ being the
  restriction of $f(q-)$ to $\setT$. For a couple $(\setQ,\setT)$,
  $\symbEpsilonMin$ in $\catAuto(\langLQT)$ turns out to be the map
  $ \setQ/{\sim_{\setT \cup\alphA\setT}} \nrightarrow (\setQ\cup
  \setQ\alphA)/{\sim_{\setT}},
  [q]\mapsto(\lcpbis(f(q-)\arrowvert_{\setT \cup
    \alphA\setT})^{-1}\lcpbis(f(q-)\arrowvert_{\setT}),[q])$, the
  first set of states being $\setQ$ quotiented by
  $\sim_{\setT \cup \alphA \setT}$, the second set of states being the
  quotient of $\setQ\cup \setQ\alphA$ by $\sim_{\setT}$. Let's
  understand the word a class $[q]$ is mapped to: with
  $\lcpbis(f(q-)\arrowvert_{\setT})$, we mean the $\lcpbis$ of the
  function $f(q-)$ restricted to the domain $\setT$, that is, the
  longest common prefix of the subset
  $\left \{f(qt)\arrowvert t \in \setT \right \}$; with
  $\lcpbis(f(q-)\arrowvert_{\setT \cup
    \alphA\setT})^{-1}\lcpbis(f(q-)\arrowvert_{\setT})$, we mean the
  word $\lcpbis(f(q-)\arrowvert_{\setT})$ from which
  $\lcpbis(f(q-)\arrowvert_{\setT \cup \alphA\setT})$ (one of its
  prefixes, as it is the longest common prefix of a bigger set of
  words) has been stripped away; when one of the $\lcpbis$s is
  undefined, $[q]$ is mapped to undefined.
	
	$\symbEpsilonMin$ is an isomorphism if and only if
$\proj2(\symbEpsilonMin)$ is a bijection and $\proj1(\symbEpsilonMin)$
is the constant function mapping every $x\in X$ to $\varepsilon$. We
can verify that $\proj2(\symbEpsilonMin)$ being a surjection is
equivalent to the condition that for all $q \in \setQ$ and $a \in
\alphA$ there exists $q'\in \setQ$ such that $q'\sim_{\setT} qa$,
whereas $\proj2(\symbEpsilonMin)$ being an injection is equivalent to
the condition that if $q$ and $q'$ are such that $q\sim_{\setT} q'$,
then $q\sim_{\setT \cup \alphA \setT} q'$. Finally, the condition on
$\proj1(\symbEpsilonMin)$ is true if and only if
$\lcpbis(f(q-)\arrowvert_{\setT \cup
\alphA\setT})=\lcpbis(f(q-)\arrowvert_{\setT})$. Remarkably, these
three naturally arising conditions turn out to be equivalent to the
ones required in Vilar's learning algorithm for "subsequential
transducers" (see \cite{Vilar96}).
\end{example}

Every time the while cycle runs, our algorithm adds either all words
$\setQ\alphA$ to $\setQ$ or all words $\alphA\setT$ to $\setT$: this
is not strictly necessary. We show next that it is sufficient to add
just one properly chosen single word  $qa\in \setQ\alphA$ or
$at\in \alphA\setT$, preserving the correctness of the algorithm.  The
canonical inclusion
$\coprod_{\setQ}\langL(\objIn)\to\coprod_{\setQ\cup\{qa\}}\langL(\objIn)$
induces a canonical morphism between the factorizations
$\factQT \rightarrowtail \factQT[\setQ \cup \left \{qa \right
\},\setT]$. Similarly, the canonical restriction
$\prod_{\setT\cup\{at\}}\langL(\objOut)\to\prod_{\setT}\langL(\objOut)$
induces a canonical morphism between the factorizations
$\factQT \twoheadleftarrow \factQT[\setQ,\setT \cup \left \{at \right
\}]$, which will be featured in the optimized algorithm.
    
\begin{theorem}
  \label{thm:optimized-algorithm}
  Algorithm~\ref{algorithm:main} can be optimized by replacing lines~\ref{line5} and~\ref{line8} respectively by:
  \begin{itemize}
  \item \emph{add to $\setQ$ a $qa \in \setQ\alphA$
        s.t.
        $\factQT \rightarrowtail \factQT[\setQ \cup \left \{qa
        \right \},\setT]$ is not an isomorphism}; 
    \item \emph{add to $\setT$ an $at \in \alphA \setT$ s.t.
        $\factQT \twoheadleftarrow \factQT[\setQ,\setT \cup
        \left \{at \right \}]$ is not an isomorphism}.
    \end{itemize}
\end{theorem}

\section{Conclusion and future work}
\label{sec:conclusion}

In this paper, we described the abstract algorithm $\FunL$, a
categorical version of Angluin's "$\Lstar$-algorithm" for learning
word automata. The focus was on providing a minimalistic category
theoretic framework for learning, with as few assumptions as possible,
emphasizing along the way the deep connection between learning and
minimization.

So far, $\FunL$ does not cover instances of the "$\Lstar$-like
algorithms" such as nominal automata, or automata/transducers over
trees. A natural continuation is to develop these
generalizations. Another aspect to understand abstractly is the
complexity of this algorithm in terms of number of
"evaluation@evaluation queries" and "equivalence queries".

% \bibliographystyle{plain}
% \bibliography{bib}

\newpage
\appendix

\section{Missing proofs from Section~\ref{sec:core-alg}}

\textbf{Notation.} Given a coproduct $\coprod_{i\in I} X_i$ and
arrows $f_i\colon X_i\to Y$, we denote by $\coprod_{i\in I}f_i$ the
mediating arrow $\coprod_{i\in I} X_i\to Y$ obtained from the
universal property. The notation
$\prod_{i\in I} g_i\colon X\to\prod_{i\in I}Y_i$ is defined dually for
a family of morphisms $g_i\colon X\to Y_i$.

\begin{proof}[Proof of Lemma~\ref{lem:input-cats}]
  We will show that $\catOQT$ is a subcategory of $\catO$, the fact
  that $\iStar$ is well defined is proved analogously.
  To this end, we have to check that, for every word $w$ resulting
  from the sequential concatenation of words in $\setQ,\alphA$ and
  $\setT$, there is only an arrow $\symbLeft w \symbRight$ in
  $\catOQT$.  We will use the "coherence diagrams" in the definition
  of the category $\catIQT$, see Definition~\ref{def:catIQT}.

  Suppose we have two ways of getting $w$ as above, $qat$ and
  $q'a't'$: let's check that
  $\symbLeftQ\symbAQT\symbRightT=\symbLeftQ[q']\symbAQT[a']\symbRightT[t']$.

  Assume $q=q_1\ldots q_n$, $q'=q'_1\ldots q'_{n'}$,
  $t=t_1\ldots t_m$, $t'=t'_1\ldots t'_{m'}$, where
  $q_i,q'_i,t_i,t'_i\in\alphA$ are letters.

  Without loss of generality, suppose that $q$ is a prefix of $q'$; by
  induction, we get
  $(\symbLeftQ[q_1\ldots q_n])\symbAQT (\symbRightT[t_1\ldots t_m]) =
  (\symbLeftQ[q_1\ldots
  q_na])\symbAQT[\varepsilon](\symbRightT[t_1\ldots t_m]) =
  (\symbLeftQ[q_1\ldots q_na])\symbAQT[t_1](\symbRightT[t_2\ldots
  t_m]) =\newline (\symbLeftQ[q_1\ldots
  q_nat_1])\symbAQT[\varepsilon](\symbRightT[t_2\ldots t_m]) = \ldots
  = (\symbLeftQ[q'_1\ldots
  q'_{n'}])\symbAQT[a'](\symbRightT[t'_1\ldots t'_{m'}])$, as desired.

  Notice that prefix-closure of $\setQ$ and suffix-closure of
  $\setT$ have been essential in order to properly use the "coherence
  diagrams".
\end{proof}

\begin{proof}[Proof of Lemma~\ref{lem:initial-QT-biautomaton}]
  It is easy to check directly the universal properties of the two
  objects.

  Suppose we have a "$\catC_{\setQ,\setT}$-biautomaton" $\autB$:
  \begin{center}
    \begin{tikzcd}
      \langL(\objIn) \arrow[rightarrow]{r}{\symbLeftQ} & B_1
      \arrow[rightarrow, shift left]{r}{\symbAQT}
      \arrow[rightarrow, shift right]{r}[swap]{\symbEpsilon} & B_2
      \arrow[rightarrow]{r}{\symbRightT} & \langL(\objOut).
    \end{tikzcd}
  \end{center}

  We define a morphism from $\autoInit(\langL)$ to $\autB$ as
  follows.
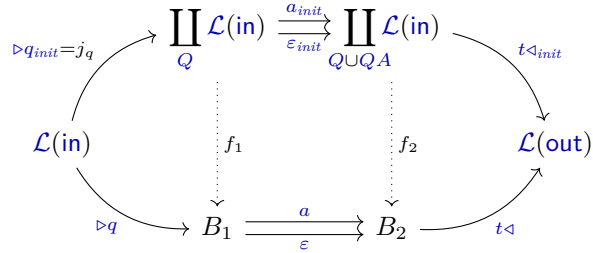
\begin{center}
\begin{tikzcd}[column sep=2em, row sep=1.5em]
  & \displaystyle\coprod_{\mathclap{\setQ}}\langL(\objIn)
  \arrow[rightarrow,dotted]{dd}{f_1} \arrow[rightarrow, shift
  left]{r}{\symbAQTInit} \arrow[rightarrow, shift
  right]{r}[swap]{\symbEpsilonInit} &
  \displaystyle\coprod_{\mathclap{\setQ\cup \setQ\alphA}}\langL(\objIn) \arrow[rightarrow, dotted]{dd}{f_2} \arrow[rightarrow, bend left]{rd}{\symbRightTInit} & \\
  \langL(\objIn) \arrow[rightarrow, bend left]{ru}{\symbLeftQInit=j_q} \arrow[rightarrow, bend right]{rd}[ swap]{\symbLeftQ}& & &\langL(\objOut)\\
  & B_1 \arrow[rightarrow, shift left]{r}{\symbAQT}
  \arrow[rightarrow, shift right]{r}[swap]{\symbEpsilon} & B_2
  \arrow[rightarrow, bend right]{ru}[swap]{\symbRightT}
\end{tikzcd}
\end{center}

The morphism $f_1$ is defined using the universal property of
$\coprod_{\setQ}\langL(\objIn)$ as the map
$\coprod_{q\in\setQ}\autB(\symbLeftQ)$. Similarly, $f_2$ is
defined using the universal property of
$\coprod_{\setQ\cup\setQ\alphA}\langL(\objIn)$ and the morphisms
$\autB(\symbEpsilon\circ\symbLeftQ)$ and $\autB(\symbAQT\circ\symbLeftQ)$.

Since $f_1$ and $f_2$ are defined using the universal property of the
two coproducts, they must be unique.

In addition, this is a natural transformation: the triangle on the
left commutes because of the definition of $f_1$; if we denote by
$j_q$ the coproduct injections, the square in the middle commutes
because
$f_2 \circ \symbAQTInit=\symbAQT \circ f_1 \Leftrightarrow f_2 \circ
\symbAQTInit \circ j_q=\symbAQT \circ f_1\circ j_q\;\;\forall q \in
\setQ$ and
$f_2 \circ \symbEpsilonInit=\symbEpsilon \circ f_1 \Leftrightarrow f_2
\circ \symbEpsilonInit \circ j_q=\symbEpsilon \circ f_1 \circ
j_q\;\;\forall q \in \setQ$; the triangle on the right commutes
because
$\symbRightT \circ f_2=\symbRightTInit \Leftrightarrow \symbRightT
\circ f_2 \circ j_{\widetilde{q}}=\symbRightTInit \circ
j_{\widetilde{q}}\;\;\forall \widetilde{q} \in \setQ \cup
\setQ\alphA$,
$\langL(\symbLeftQ \symbAQT \symbRightT)=\langL(\symbLeftQInit
\symbAQTInit \symbRightTInit)$ and
$\langL(\symbLeftQ \symbRightT)=\langL(\symbLeftQInit
\symbRightTInit)$.

The reasoning for the final object is perfectly dual.
\end{proof}

\begin{proof}[Proof of Lemma~\ref{consistencypreserved}]
  We have to prove that for all $q\in\setQ$, $a\in\alphA$ and
  $t\in\setT$, we have
  \begin{itemize}
  \item $\hypQT(\symbLeft q a t\symbRight)=\langL(\symbLeft q a t\symbRight)$, and
  \item $\hypQT(\symbLeft q  t\symbRight)=\langL(\symbLeft q  t\symbRight)$.
  \end{itemize}
  We prove the first item, as the second is similar.  Consider a word
  $w\in\alphA^*$, with $w=a^1\ldots a^n$ where $a^i\in\alphA$. Upon
  seeing $w$ as a morphism $w\colon\objStates\to\objStates$ in $\catI$, we know that
   $\hypQT(w)$ is by definition the following composite.
  \begin{center}
    \begin{tikzcd}
      \objMin(\langLQT)(\objBiStates1)\ar[r,"a^1_{\mathit{min}}"] &
      \objMin(\langLQT)(\objBiStates2)\ar[r,"\symbEpsilonMin^{-1}"] &
      \ldots\ar[r,"a^n_{\mathit{min}}"] &
      \objMin(\langLQT)(\objBiStates2)\ar[r,"\symbEpsilonMin^{-1}"] &
      \objMin(\langLQT)(\objBiStates1)
    \end{tikzcd}
  \end{center}
  We can prove the following.
  \begin{itemize}
  \item If $w\in\setQ$, then $\hypQT(w)\circ\symbLeftQMin[\varepsilon]=\symbLeftQMin[w]$.
  \item If $w\in\setT$, then
    $\symbRightTMin[\varepsilon]\circ\symbEpsilonMin\circ
    \hypQT(w)=\symbRightTMin[w]\circ\symbEpsilonMin$.
  \end{itemize}
  These statements follow easily by applying repeatedly the "coherence
  diagrams" from the definition of $\catIQT$ and using the
  prefix-closure of $\setQ$ and the suffix-closure of $\setT$.  To
  improve readability, we simply write $\mathcal{H}$ for $\hypQT$ in
  the next equalities, obtained by applying the previous equations:
  \begin{align*}
    \hypQT(\symbLeft qat\symbRight)  & = (\symbRightTMin[\varepsilon]\circ\symbEpsilonMin\circ \mathcal{H}(t))\circ \mathcal{H}(a)\circ \mathcal{H}(q) \circ\symbLeftQMin[\varepsilon]\\
                                     & =(\symbRightTMin\circ\symbEpsilonMin)\circ \mathcal{H}(a) \circ \mathcal{H}(q) \circ\symbLeftQMin[\varepsilon]\\
                                     & =\symbRightTMin\circ\symbEpsilonMin\circ(\symbEpsilonMin^{-1}\circ\symbAQTMin) \circ \mathcal{H}(q) \circ\symbLeftQMin[\varepsilon]\\
                                             & =\symbRightTMin\circ\symbAQTMin \circ (\mathcal{H}(q) \circ\symbLeftQMin[\varepsilon])\\
                                     & =\symbRightTMin\circ\symbAQTMin  \circ\symbLeftQMin[q]\\
                                     &=\objMin(\langLQT)(\symbLeftQ a \symbRightT)\\
    &=\langL(\symbLeft qat\symbRight).
  \end{align*}
  Proving
  $\hypQT(\symbLeft q t\symbRight)=\langL(\symbLeft q t\symbRight)$
  follows the same lines as above.
\end{proof}

\begin{proof}[Proof of Example~\ref{example:finite-is-noetherian}]
  Only the case of "$(\EpiKlT,\MonoKlT)$-noetherian" has to be proved.

  Let us recall first that an isomorphism in~$\KlT$ is an
  arrow~$f\colon X\nrightarrow Y$ such that~$\proj1(f)$ is the constant function equal to~$\varepsilon$ and~$\proj2(f)$ is a bijection.

  Let us assume that~$X$ is finite and show that it
  is~"$(\EpiKlT,\MonoKlT)$-noetherian".  For this, consider two
  commutative diagrams as follows:
  \begin{center}
    \begin{tikzcd}
      &                                   & X \arrow[lld, "f_1",two heads, bend right] \arrow[ld, "f_2", two heads, bend right] \arrow[two heads, bend right,color=white]{d}[color=black]{\ldots} &  &                                                           &                                                          & X                      \\
      Y_1 & Y_2 \arrow[l, "e_1", two heads] & \ldots \arrow[l, "e_2",
      two heads] & & Z_1 \arrow[r, "m_1", tail] \arrow[rru, "g_1",
      tail, bend left] & Z_2 \arrow[r, "m_2", tail] \arrow[ru, "g_2",
      tail, bend left] & \ldots \arrow[bend right,color=white,
      tail]{u}[color=black]{\ldots}
    \end{tikzcd}
  \end{center}
  Concerning the left diagram, recall first that for all
  $e\colon Y\epirightarrow Z$ from~$\EpiKlT$, $\proj2(e)$ is
  surjective. Hence, if~$Y$ is finite, then so is~$Z$ and
  furthermore~$|Y|\geqslant|Z|$.  Applied to our diagram, we obtain
  that $(|Y_i|)_i$ is a non-decreasing sequence of non-negative
  integers bounded by~$|X|$.  Hence, for all~$i$ sufficiently large
  (say, larger than~$i_0$), $|Y_i|=|Y_{i+1}|$ and as a consequence
  $\proj2(e_i)$ is a totally defined bijection.

  Given an arrow~$f\colon Y\nrightarrow Z$, denote by~$||f||$ the sum
  of the~$|\proj1(f)(y)|$ for~$y\in Y$ (the sum of the sizes of all output
  words).  Consider now some~$g\colon Z\nrightarrow T$.  It is easy to
  see that if~$\proj2(f)$ is surjective and $\proj2(g)$ is totally
  defined, then $||g\circ f||\geqslant||g||+||f||$. It follows that
  $||f_{i}||\geqslant ||f_{i+1}||+||e_i||$ for all~$i\geqslant i_0$.
  Since~$||f_{i_0}||$ is finite, it follows that $||e_i||=0$ for
  all~$i$ sufficiently large.  Hence, $e_i$ is an isomorphism: we have
  established the first item of \cref{definition:EM-noetherian}.

  Concerning the right diagram, note that for all
  $m\colon Y\monorightarrow Z$ from~$\MonoKlT$, $\proj2(m)$ is totally
  defined and injective. Hence, if~$Z$ is finite, $Y$ also is and
  $|Y|\leqslant|Z|$. It follows in our case that $(|Z_i|)_i$ is a non
  decreasing sequence, bounded by~$|X|$.  Hence, for all $i$
  sufficiently large, $|Z_i|=|Z_{i+1}|$. It follows that~$\proj2(m_i)$
  is an injection from~$Z_i$ to~$Z_{i+1}$ with $|Z_i|=|Z_{i+1}|$
  finite. Thus~$\proj2(m_i)$ is a bijection. Since
  furthermore~$\proj1(m_i)=\varepsilon$ by assumption, we obtain
  that~$m_i$ is an isomorphism: we have established the second item of
  \cref{definition:EM-noetherian}.

  Overall~$X$ is "$(\EpiKlT,\MonoKlT)$-noetherian".

  For the converse direction, let us consider a
  "$(\EpiKlT,\MonoKlT)$-noetherian" set~$X$ to show that it is finite. For
  the sake of contradiction, assume it would be infinite, and let
  $Z_0\subsetneq Z_1\subsetneq Z_2 \subsetneq \cdots$ be a sequence of
  subsets of~$Z$.  Let now~$m_i\colon Z_i\monorightarrow Z_{i+1}$ be
  the inclusion from~$Z_i$ into~$Z_{i+1}$, and
  $f_i\colon Z_i \monorightarrow X$ be the inclusion of~$Z_i$
  into~$X$. These arrows show that the second item of
  \cref{definition:EM-noetherian} does not hold.
\end{proof}

\section{Termination and correctness of the \FunL{-}algorithm}

\begin{definition}
  \label{def:factQT}
  For all couples $(\setQ,\setT)$, the language $\langL$ induces a
  canonical map

\[\displaystyle\coprod_{q\in \setQ}\prod_{t\in\setT}\langL(\symbLeft q
  t\symbRight)\colon\coprod_{\setQ}\langL(\objIn)\to\prod_{\setT}\langL(\objOut)\,,\]
obtained using the universal properties of the coproduct, respectively
of the product.  We define $\intro\factQT$ as the factorization of
this map
\begin{center}
  \begin{tikzcd}
    \coprod_{\setQ}\langL(\objIn) \arrow[r, two heads] & \factQT
    \arrow[r, tail] & \prod_{\setT}\langL(\objOut).
  \end{tikzcd}
\end{center}
\end{definition}

The objects $\factQT$ will play a crucial role in the proof of
termination of the \FunL{-}algorithm. First notice that, when
$\setQ=\setT=\alphA^*$, then $\fact{\alphA^*,\alphA^*}$ is nothing
else but the state space of the "minimal automaton" $\objMin(\langL)$,
as computed in Theorem~\ref{minimalwordauto}. 
More generally, for an
arbitrary $(\setQ,\setT)$, objects of the form $\factQT$ are featured
in the minimal "$\catC_{\setQ,\setT}$-biautomaton"
$\objMin(\langLQT)$, whose computation from
Theorem~\ref{thm:minimal-QT-biautomaton} is recalled below.
\begin{center}
  \begin{tikzcd}[column sep=3em]
    & \displaystyle\coprod_{\mathclap{\setQ}}\langL(\objIn)
    \arrow[twoheadrightarrow]{d}{\epiOne} \arrow[rightarrow,
    shift left]{r}[yshift=0.2ex]{\symbAQTInit} \arrow[rightarrow,
    shift right]{r}[swap, yshift=-0.4ex]{\symbEpsilonInit} &
    \displaystyle\coprod_{\mathclap{\setQ\cup \setQ\alphA}}\langL(\objIn) \arrow[twoheadrightarrow]{d}{\epiTwo} \arrow[rightarrow, bend left]{rd}{\symbRightTInit} & \\
    \langL(\objIn) \arrow[rightarrow]{r}{\symbLeftQMin} \arrow[rightarrow, bend left]{ru}{\symbLeftQInit} \arrow[rightarrow, bend right]{rd}[swap]{\symbLeftQFinal}& \objMin(\langLQT)(\objBiStates1) \arrow[rightarrowtail]{d}{\monoOne} \arrow[rightarrow, shift left]{r}[yshift=0.2ex]{\symbAQTMin} \arrow[rightarrow, shift right]{r}[swap,yshift=-0.4ex]{\symbEpsilonMin} &\objMin(\langLQT)(\objBiStates2)\arrow[rightarrow]{r}{\symbRightTMin} \arrow[rightarrowtail]{d}{\monoTwo} &\langL(\objOut)\\
    & \displaystyle\prod_{\mathclap{\setT \cup \alphA
        \setT}}\langL(\objOut) \arrow[rightarrow, shift
    left]{r}[yshift=0.2ex]{\symbAQTFinal} \arrow[rightarrow, shift
    right]{r}[swap, yshift=-0.4ex]{\symbEpsilonFinal} &
    \displaystyle\prod_{\mathclap{\setT}}\langL(\objOut)
    \arrow[rightarrow, bend right]{ru}[swap]{\symbRightTFinal} &
  \end{tikzcd}
\end{center}

Indeed, a careful analysis of the situation described in the above
diagram reveals that when computing the unique map from
$\autoInit(\langLQT)$ to $\autoFinal(\langLQT)$ we obtain that:
\begin{itemize}
\item the unique map
  $\autoInit(\langLQT)(\objBiStates1)\to\autoFinal(\langLQT)(\objBiStates1)$
  is the canonical map
  \[\displaystyle\coprod_{q\in
      \setQ}\prod_{t\in\setT\cup\alphA\setT}\langL(\symbLeft q
    t\symbRight)\colon\coprod_{\setQ}\langL(\objIn)\to\prod_{\setT\cup\alphA\setT}\langL(\objOut)\,;\]
\item the unique map
  $\autoInit(\langLQT)(\objBiStates2)\to\autoFinal(\langLQT)(\objBiStates2)$
  is the canonical map
  \[\displaystyle\coprod_{q\in
      \setQ\cup\setQ\alphA}\prod_{t\in\setT}\langL(\symbLeft q
    t\symbRight)\colon\coprod_{\setQ\cup\setQ\alphA}\langL(\objIn)\to\prod_{\setT}\langL(\objOut)\,;\]
\item the composite map $\monoTwo\circ\symbEpsilonMin\circ\epiOne$ is
  the canonical map
  \[\displaystyle\coprod_{q\in \setQ}\prod_{t\in\setT}\langL(\symbLeft
    q
    t\symbRight)\colon\coprod_{\setQ}\langL(\objIn)\to\prod_{\setT}\langL(\objOut)\,.\]
\end{itemize}

We therefore obtain the following lemma.

\begin{lemma}
  \label{lem:factQT-and-minimal-automaton} The states of the minimal
  "$\catC_{\setQ,\setT}$-biautomaton" $\objMin(\langLQT)$ are
  $\fact{\setQ,\setT\cup\alphA\setT}$ and, respectively,
  $\fact{\setQ\cup\setQ\alphA,\setT}$.
\end{lemma}

The next lemma lists some simple, yet very useful properties of the
objects $\factQT$.

\begin{lemma}
  \label{lem:factQT-properties}
  We assume $\setQ,\setQ',\setT,\setT'$ are subsets of $\alphA^*$ such
  that $\setQ\subseteq\setQ'$ and $\setT\subseteq\setT'$. Then the
  following properties hold.
  \begin{enumerate}
  \item\label{it:can-mono} There exists a unique canonical
    $\Mono$-morphism
    \begin{tikzcd}
      \factQT\ar[r,tail] &\fact{\setQ',\setT}
    \end{tikzcd} such that the following diagram commutes, in which
    $\intro*\can$ denotes the canonical morphism between the two
    coproducts.
    \begin{center}
      \begin{tikzcd}
        \coprod_{\setQ}\langL(\objIn) \arrow[r, two heads] \arrow[d, "\can"] & \factQT \arrow[r, tail] \arrow[d, dotted, tail] & \prod_{\setT}\langL(\objOut) \arrow[transform canvas={xshift=0.3ex},-]{d} \arrow[transform canvas={xshift=-0.3ex},-]{d} \\
        \coprod_{\setQ'}\langL(\objIn) \arrow[r, two heads] &
        {\fact{\setQ',\setT}} \arrow[r, tail] &
        \prod_{\setT}\langL(\objOut)
      \end{tikzcd}
    \end{center}
  \item\label{it:can-epi} There exists a unique canonical
    $\Epi$-morphism
    \begin{tikzcd}
      \fact{\setQ,\setT'}\ar[r,two heads] &\factQT
    \end{tikzcd} such that the following diagram commutes, in which
    $\reintro*\can$ denotes the canonical morphism between the two
    products.
    \begin{center}
      \begin{tikzcd}
        \coprod_{\setQ}\langL(\objIn) \arrow[r, two heads] \arrow[transform canvas={xshift=0.3ex},-]{d} \arrow[transform canvas={xshift=-0.3ex},-]{d} & \fact{\setQ,\setT'} \arrow[r, tail] \arrow[d, dotted, two heads] & \prod_{\setT'}\langL(\objOut) \arrow[d, "\can"] \\
        \coprod_{\setQ}\langL(\objIn) \arrow[r, two heads] &
        {\fact{\setQ,\setT}} \arrow[r, tail] &
        \prod_{\setT}\langL(\objOut)
      \end{tikzcd}
    \end{center}
  \item\label{it:can-comp} For these canonical morphisms, the
    following diagram commutes:
    \begin{center}
      \begin{tikzcd}
        {\fact{\setQ,\setT'}} \arrow[d, two heads] \arrow[r, tail] & {\fact{\setQ',\setT'}} \arrow[d, two heads] \\
        {\fact{\setQ,\setT}} \arrow[r, tail] & {\fact{\setQ',\setT}}.
      \end{tikzcd}
    \end{center}
  \item\label{it:can-mono-iso} Furthermore, if
    \begin{tikzcd}
      {\fact{\setQ,\setT'}} \arrow[r, tail] & {\fact{\setQ',\setT'}}
    \end{tikzcd}
    is an iso, then so is
    \begin{tikzcd}
      {\fact{\setQ,\setT}} \arrow[r, tail] & {\fact{\setQ',\setT}}
    \end{tikzcd}.
  \item\label{it:can-epi-iso} And dually, if
    \begin{tikzcd}
      {\fact{\setQ',\setT'}} \arrow[r, two heads] &
      {\fact{\setQ',\setT}}
    \end{tikzcd}
    is an iso, then so is
    \begin{tikzcd}
      {\fact{\setQ,\setT'}} \arrow[r, two heads] &
      {\fact{\setQ,\setT}}.
    \end{tikzcd}
  \item\label{it:can-comp-well} The canonical morphisms compose well,
    that is, for $\setQ'\subseteq\setQ''$ and
    $\setT'\subseteq\setT''$, we have
    \begin{center}
      \begin{tikzcd}
        {\fact{\setQ,\setT''}} \arrow[r, two heads] \arrow[rr,two
        heads,bend left] & {\fact{\setQ,\setT'}} \arrow[r, two heads]
        & {\fact{\setQ,\setT}} & \text{and} & {\fact{\setQ,\setT}}
        \arrow[r, tail] \arrow[rr, tail, bend left] &
        {\fact{\setQ',\setT}} \arrow[r, tail] &
        {\fact{\setQ'',\setT}}.
      \end{tikzcd}
    \end{center}
  \end{enumerate}
\end{lemma}

\begin{proof}
  
  \begin{enumerate}
  \item The existence and unicity of the morphism follows from the
    diagonal fill-in property of the factorization system
    $(\Epi,\Mono)$. The fact that it is in $\Mono$ follows from the
    cancellation property of $\Mono$, see, e.g.~\cite[Proposition
    14.9(2)]{joyofcats}.
  \item The proof is dual to that of the previous item.
  \item Notice that, by the diagonal fill-in property, we have a
    unique canonical morphism
    \begin{tikzcd}
      \fact{\setQ,\setT'}\ar[r] &\fact{\setQ',\setT}
    \end{tikzcd}
    such that the following diagram commutes, in which the vertical
    arrows denote the canonical morphisms between the two coproducts,
    respectively between the two products.
    \begin{center}
      \begin{tikzcd}
        \coprod_{\setQ}\langL(\objIn) \arrow[r, two heads] \arrow[d, "\can"] & \fact{\setQ,\setT'} \arrow[r, tail] \arrow[d, dotted,"\exists!"] & \prod_{\setT'}\langL(\objOut) \arrow[d, "\can"] \\
        \coprod_{\setQ'}\langL(\objIn) \arrow[r, two heads] &
        {\fact{\setQ',\setT}} \arrow[r, tail] &
        \prod_{\setT}\langL(\objOut)
      \end{tikzcd}
    \end{center}
    The commutativity of the diagram follows since both
    \begin{tikzcd}
      {\fact{\setQ,\setT'}} \arrow[r, tail] & {\fact{\setQ',\setT'}}
      \arrow[r, two heads] & {\fact{\setQ',\setT}}
    \end{tikzcd}
    and
    \begin{tikzcd}
      {\fact{\setQ,\setT'}} \arrow[r, two heads]&{\fact{\setQ,\setT}}
      \arrow[r, tail] & {\fact{\setQ',\setT}}
    \end{tikzcd}
    make the above diagram commute.
  \item Assume \begin{tikzcd} {\fact{\setQ,\setT'}} \arrow[r, tail] &
      {\fact{\setQ',\setT'}}
    \end{tikzcd} is an isomorphism. Then the canonical morphism
    \begin{center}
      \begin{tikzcd} {\fact{\setQ,\setT'}} \arrow[r] &
        {\fact{\setQ',\setT}}
      \end{tikzcd}
    \end{center}
    is in $\Epi$. Using the cancellative property of $\Epi$ (the dual
    of~\cite[ Proposition 14.9(2)]{joyofcats}), we obtain that the
    canonical morphism
    \begin{tikzcd} {\fact{\setQ,\setT}} \arrow[r, tail] &
      {\fact{\setQ',\setT}}
    \end{tikzcd} is in $\Epi$. Since it is already in $\Mono$, it must
    be in $\Epi\cap\Mono$, hence it is an isomorphism.
  \item The proof is dual to that of the previous item.
  \item The proof is very similar to that of item~\ref{it:can-comp}
    and relies on the uniqueness of the canonical morphisms.
  \end{enumerate}
\end{proof}

Using Lemma~\ref{lem:factQT-and-minimal-automaton} together with
items~\ref{it:can-mono} and~\ref{it:can-epi} of
Lemma~\ref{lem:factQT-properties}, we obtain that $\factQT$ is a
$(\Epi,\Mono)$-factorization of the morphism $\symbEpsilonMin$. \AP To
underline the dependence of $\symbEpsilonMin$ on $\setQ$ and $\setT$,
we will write hereafter $\intro*\symbEpsilonMinQT$ for this map.

\begin{lemma}
  \label{lem:epi-mono-fact-of-epsilon}
  The morphism $\symbEpsilonMinQT=\objMin(\langLQT)(\symbEpsilon)$ is
  the composite of the canonical morphims
  \begin{tikzcd}
    \fact{\setQ,\setT\cup\alphA\setT}\ar[r,two heads]&\factQT\ar[r,
    tail]& \fact{\setQ\cup\setQ\alphA,\setT}
  \end{tikzcd}. Consequently:
  \begin{itemize}
  \item if $\symbEpsilonMinQT\not\in\Epi$, then \begin{tikzcd}
      \factQT\ar[r, tail]& \fact{\setQ\cup\setQ\alphA,\setT}
    \end{tikzcd} is not an isomorphism;
  \item if $\symbEpsilonMinQT\not\in\Mono$ then \begin{tikzcd}
      \fact{\setQ,\setT\cup\alphA\setT}\ar[r,two heads]&\factQT
    \end{tikzcd} is not an isomorphism.
  \end{itemize}
\end{lemma}
\begin{proof}
  In a diagram, we have the following situation:
  \begin{center}
    \begin{tikzcd}
      \coprod_{\setQ}\langL(\objIn) \arrow[rr,"\symbEpsilonInit"] \arrow[d, two heads] \arrow[rd, two heads, bend right=20, in=90] &                                                     & \coprod_{\setQ\cup\setQ\alphA}\langL(\objIn) \arrow[d, two heads] \\
      {\fact{\setQ,\setT\cup\alphA\setT}} \arrow[r, two heads] \arrow[d, tail]            & \factQT \arrow[r, tail] \arrow[rd, tail, bend left=20,out=270] & {\fact{\setQ\cup\setQ\alphA,\setT}} \arrow[d, tail]              \\
      \prod_{\setT\cup\alphA\setT}\langL(\objOut)
      \arrow[rr,swap,"\symbEpsilonFinal"] & &
      \prod_{\setT}\langL(\objOut).
    \end{tikzcd}
  \end{center}
  The second part of the lemma is immediate, since $\Epi$ and $\Mono$
  contain the isomorphisms and are closed under composition.
\end{proof}

Before proving the main result of this section, we require two more
results which use the "$(\Epi,\Mono)$-noetherianity" of the state space
of the "minimal automaton" $\objMin(\langL)$.
\begin{lemma}
  \label{lem:no-inf-seq}
  The following hold.
  \begin{enumerate}
  \item There does not exist an infinite sequence of subsets
    $(\setT_n)_n$ with $\setT_n\subseteq \setT_{n+1}$ and such that
    the canonical morphisms below are not isomorphisms:
    \begin{center}
      \begin{tikzcd}
        \fact{\alphA^*,{\setT_1}} & \fact{\alphA^*,{\setT_2}}\ar[l,two
        heads] & \fact{\alphA^*,{\setT_3}} \ar[l,two heads]
        &\ldots\ar[l,two heads]
      \end{tikzcd}
    \end{center}
  \item Dually, there does not exist an infinite sequence of subsets
    $(\setQ_n)_n$ with $\setQ_n\subseteq \setQ_{n+1}$ and such that
    the canonical morphisms below are not isomorphisms:
    \begin{center}
      \begin{tikzcd}
        \fact{{\setQ_1},\alphA^*}\ar[r,tail] &
        \fact{{\setQ_2},\alphA^*}\ar[r,tail] &
        \fact{{\setQ_{3}},\alphA^*}\ar[r,tail] & \ldots
      \end{tikzcd}
    \end{center}
  \end{enumerate}
\end{lemma}

\begin{proof}
  Notice that, by item~\ref{it:can-comp-well} of
  Lemma~\ref{lem:factQT-properties}, we have the following commuting
  diagram of canonical $\Epi$-quotients:
  \begin{center}
    \begin{tikzcd}
      & & & \fact{\alphA^*,\alphA^*}\ar[dlll,bend right=25,two heads]\ar[dll,bend right=20,two heads]\ar[dl,bend right=10,two heads] \arrow[color=white]{d}[color=black, description]{\ldots}\\
      \fact{\alphA^*,{\setT_1}} & \fact{\alphA^*,{\setT_2}}\ar[l,two
      heads] & \fact{\alphA^*,{\setT_3}} \ar[l,two heads]
      &\ldots\ar[l,two heads]
    \end{tikzcd}
  \end{center}
  We conclude that the sequence cannot exist upon recalling that
  $\fact{\alphA^*,\alphA^*}$ is isomorphic to the state space of the
  "minimal automaton", and hence it is "$(\Epi,\Mono)$-noetherian". By
  the first property from Definition~\ref{definition:EM-noetherian},
  such a sequence cannot exist.
  
  The proof for the second part is dual, as now we have an infinite
  sequence of $\Mono$-subobjects of $\fact{\alphA^*,\alphA^*}$ and we
  can use the second property from
  Definition~\ref{definition:EM-noetherian}.
\end{proof}

In proving the termination of the while loop, we will consider
sequences of couples $(\setQ_i,\setT_i)_{i\ge 0}$ which can be added
either because $\symbEpsilonMinQT[\setQ_i,\setT_i]$ is not in $\Epi$
or because it is not in $\Mono$, obtaining an $\Epi$-quotient
\begin{tikzcd}\fact{\setQ_i,\setT_i} &
  \fact{\setQ_{i+1},\setT_{i+1}}\ar[l,two heads,"e_i"]
\end{tikzcd}
or a $\Mono$-subobject
\begin{tikzcd}
  \fact{\setQ_i,\setT_i}\ar[r,
  tail,"m_i"]&\fact{\setQ_{i+1},\setT_{i+1}}
\end{tikzcd}.

For this reason, the next lemma will prove helpful.

\begin{lemma}
  \label{lem:finitechain}
  Consider a possibly infinite sequence of couples
  $(\setQ_i,\setT_i)_{i\ge 1}$ related by morphisms $(f_i)_{i\ge 1}$
  such that for all $i$ either
  \begin{itemize}
  \item $\setQ_{i}\subseteq\setQ_{i+1}$, $\setT_i = \setT_{i+1}$ and
    $f_i$ is the canonical morphism
    \begin{tikzcd}\fact{\setQ_{i},\setT_{i}}
      \arrow[r,tail,"m_i"] & \fact{\setQ_{i+1},\setT_{i+1}}
    \end{tikzcd}, or
  \item $\setQ_{i+1}=\setQ_i$, $\setT_i\subseteq \setT_{i+1}$ and
    $f_i$ is the canonical morphism
    \begin{tikzcd}\fact{\setQ_{i},\setT_{i}} & \arrow[l,two
      heads,swap, "e_i"] \fact{\setQ_{i+1},\setT_{i+1}}
    \end{tikzcd}.
  \end{itemize}
  Then there can only be finitely many morphisms in $(f_i)_{i\ge 1}$
  not being isomorphisms.
\end{lemma}

  \begin{proof}
    Using items~\ref{it:can-comp} and~\ref{it:can-comp-well} of
    Lemma~\ref{lem:factQT-properties} for each index $i$, we have that
    the morphism $f_i$ is at the bottom of exactly one of the next
    diagrams. Notice, however, that for the diagram on the right,
    since we are in the case when $\setQ_{i}=\setQ_{i+1}$, the
    canonical morphism at the top of the diagram is actually the
    identity.
    \begin{center}
      \begin{tikzcd}
        \fact{\setQ_{i},\alphA^*}\ar[r,tail,"\overline{m_i}"]\ar[d,two heads] & \fact{\setQ_{i+1},\alphA^*}\ar[d,two heads] & & \fact{\setQ_{i},\alphA^*}\ar[d,two heads] & \fact{\setQ_{i+1},\alphA^*}\arrow[transform canvas={yshift=0.3ex},-]{l} \arrow[transform canvas={yshift=-0.3ex},-]{l}\ar[d,two heads]\\
        \fact{\setQ_{i},\setT_{i}}\ar[r,tail,swap,"m_i"] &
        \fact{\setQ_{i+1},\setT_{i+1}} & & \fact{\setQ_{i},\setT_{i}}
        & \fact{\setQ_{i+1},\setT_{i+1}}\ar[l,two heads,"e_i"]
      \end{tikzcd}
    \end{center}
    We therefore obtain a sequence
    \begin{center}
      \begin{tikzcd}
        \fact{\setQ_{1},\alphA^*} \ar[r,tail,"\overline{f_1}"] &
        \fact{\setQ_{2},\alphA^*} \ar[r,tail,"\overline{f_2}"] &
        \fact{\setQ_{3},\alphA^*} &\ldots
      \end{tikzcd}
    \end{center}
    where each $\overline{f_i}$ is either the morphism
    $\overline{m_i}$ from the left diagram above or the identity in
    the second case. Notice additionally that, using
    item~\ref{it:can-mono-iso} of Lemma~\ref{lem:factQT-properties},
    if $m_i$ is not an isomorphism, then $\overline{m_i}$ is not
    either: otherwise, we would get that $m_i$ is an isomorphism,
    contradicting the hypothesis.

    By Lemma~\ref{lem:no-inf-seq}, the sequence $(\overline{f_i})_i$
    can contain only finitely many morphisms of the form
    $\overline{m_i}$ not being isomorphisms.

    In a completely dual manner, we can prove that only for finitely
    many indexes $i$ the morphisms $f_i$ are of the form $e_i$ not
    being isomorphisms. Indeed, in this case we would exhibit a
    sequence of the form
    \begin{center}
      \begin{tikzcd}
        \fact{\alphA^*,\setT_{1}} & \fact{\alphA^*,\setT_{2}}
        \ar[l,two heads,"\widetilde{f_1}"] & \fact{\alphA^*,\setT_{3}}
        \ar[l,two heads,"\widetilde{f_2}"] & &\ldots
      \end{tikzcd}
    \end{center}
    where each $\widetilde{f_i}$ is either the identity or a canonical
    morphism $\widetilde{e_i}$, i.e. a lifting of $e_i$, which by
    virtue of item~\ref{it:can-epi-iso} of
    Lemma~\ref{lem:factQT-properties} cannot be an isomorphism if
    $e_i$ is not an isomorphism. We conclude just as above using
    Lemma~\ref{lem:no-inf-seq}.

    Since only finitely many $f_i$s are of the form $m_i$ not being
    isomorphisms and only finitely many of them are of the form $e_i$
    not being isomorphisms, we conclude that there can only be
    finitely many morphisms not being isomorphisms in the whole
    sequence.
  \end{proof}

  We are now ready to prove the main result of this section.
  \begin{proposition}
    \label{prop:term-while-loop}
    The \textbf{while} loop on line~\ref{line3} of
    Algorithm~\ref{algorithm:main} terminates.
  \end{proposition}

  \begin{proof}
    Consider a possibly infinite sequence $(\setQ_i,\setT_i)_{i\ge 1}$ obtained
    via a run of the \textbf{while} loop. Notice that for each $i$ the
    couple $(\setQ_{i+1},\setT_{i+1})$ was obtained from
    $(\setQ_i,\setT_i)$ either because
    \begin{itemize}
    \item[a)] $\symbEpsilonMinQT[\setQ_i,\setT_i]\not\in\Epi$, or
    \item[b)] $\symbEpsilonMinQT[\setQ_i,\setT_i]\not\in\Mono$.
    \end{itemize}
    Therefore, we know that, in the respective cases, we have either
    \begin{itemize}
    \item[a)] $\setQ_{i+1}=\setQ_{i}\cup\setQ_{i}\alphA$ and $\setT_{i+1}=\setT_{i}$, or
    \item[b)] $\setQ_{i+1}=\setQ_{i}$ and $\setT_{i+1}=\setT_{i}\cup\alphA\setT_{i}$.
    \end{itemize}
    Let $f_i$ denote exactly one the following canonical morphisms,
    depending on the respective case:
    \begin{itemize}
    \item[a)] \begin{tikzcd}
  \fact{\setQ_i,\setT_i}\ar[r,
  tail,"m_i"]&\fact{\setQ_{i+1},\setT_{i+1}}
\end{tikzcd}, or
    \item[b)]     \begin{tikzcd}\fact{\setQ_i,\setT_i} &
  \fact{\setQ_{i+1},\setT_{i+1}}\ar[l,two heads,"e_i"].
\end{tikzcd}
    \end{itemize}
    Notice that by Lemma~\ref{lem:epi-mono-fact-of-epsilon} none of
    the morphisms $f_i$ is an isomorphism.

    Therefore, we can apply Lemma~\ref{lem:finitechain} to conclude
    that the sequence $(\setQ_i,\setT_i)_{i\ge 1}$ is finite.
\end{proof}

\begin{lemma}
  \label{algotermination} Only a finite number of "counterexamples"
  with their prefixes can be added to $\setQ$.
\end{lemma}
\begin{proof}
  The situation is as follows:
  \begin{itemize}
  \item the learner asks an "equivalence query" for an
    "$\langL$-automatable" couple $(\setQ,\setT)$;
  \item the teacher answers the query negatively and provides a
    "counterexample" $w \in \alphA^*$;
  \item \AP the learner adds the counterexample and all its prefixes
    to $\setQ$, then he runs the while cycle and obtains a new
    "$\langL$-automatable" couple
    $(\intro*\setQafter,\intro*\setTafter)$, with $w$ and its prefixes
    still belonging to $\setQafter$.
  \end{itemize}

  We want to show that the complexity of the produced "hypothesis
  automata" strictly increases every time we add a "counterexample"
  and its prefixes to $\setQ$.

  \AP As $(\setQ,\setT)$ and $(\setQafter,\setTafter)$ are
  "automatable", we may suppose the fundamental commuting diagrams of
  the categories $\catAuto(\langLQT)$ and
  $\catAuto(\langLQT[\setQafter,\setTafter])$ to be as follows.

  
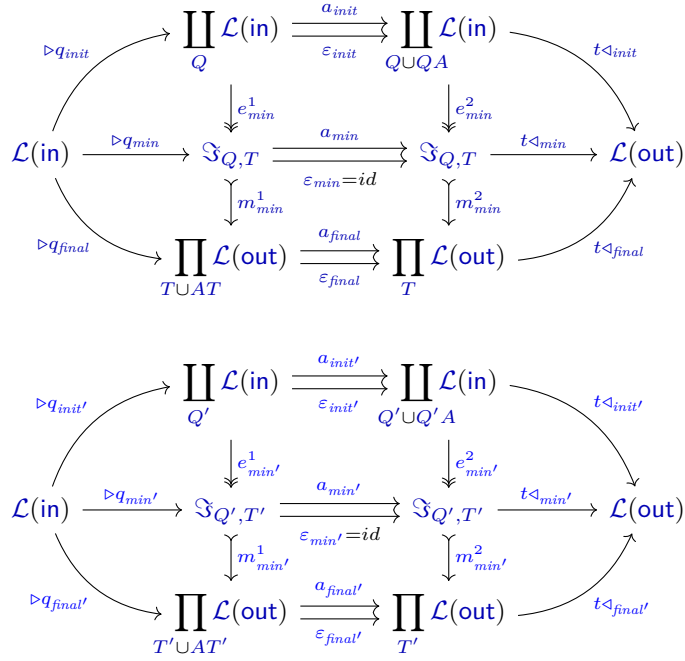
\begin{center}
    \begin{tikzcd}[column sep=3em]
      & \displaystyle\coprod_{\mathclap{\setQ}}\langL(\objIn)
      \arrow[twoheadrightarrow]{d}{\epiOne} \arrow[rightarrow, shift
      left]{r}[yshift=0.2ex]{\symbAQTInit} \arrow[rightarrow, shift
      right]{r}[swap, yshift=-0.4ex]{\symbEpsilonInit} &
      \displaystyle\coprod_{\mathclap{\setQ\cup \setQ\alphA}}\langL(\objIn) \arrow[twoheadrightarrow]{d}{\epiTwo} \arrow[rightarrow, bend left]{rd}{\symbRightTInit} & \\
      \langL(\objIn) \arrow[rightarrow]{r}{\symbLeftQMin} \arrow[rightarrow, bend left]{ru}{\symbLeftQInit} \arrow[rightarrow, bend right]{rd}[swap]{\symbLeftQFinal}& \factQT \arrow[rightarrowtail]{d}{\monoOne} \arrow[rightarrow, shift left]{r}[yshift=0.2ex]{\symbAQTMin} \arrow[rightarrow, shift right]{r}[swap,yshift=-0.4ex]{\symbEpsilonMin=id} &\factQT \arrow[rightarrow]{r}{\symbRightTMin} \arrow[rightarrowtail]{d}{\monoTwo} &\langL(\objOut)\\
      & \displaystyle\prod_{\mathclap{\setT \cup \alphA
          \setT}}\langL(\objOut) \arrow[rightarrow, shift
      left]{r}[yshift=0.2ex]{\symbAQTFinal} \arrow[rightarrow, shift
      right]{r}[swap, yshift=-0.4ex]{\symbEpsilonFinal} &
      \displaystyle\prod_{\mathclap{\setT}}\langL(\objOut)
      \arrow[rightarrow, bend right]{ru}[swap]{\symbRightTFinal} &
    \end{tikzcd}
  \end{center}
  \begin{center}
    \begin{tikzcd}[column sep=3em]
      & \displaystyle\coprod_{\mathclap{\setQafter}}\langL(\objIn)
      \arrow[twoheadrightarrow]{d}{\intro*\epiOneTr}
      \arrow[rightarrow, shift
      left]{r}[yshift=0.2ex]{\intro*\symbAQTInitTr} \arrow[rightarrow,
      shift right]{r}[swap, yshift=-0.4ex]{\intro*\symbEpsilonInitTr}
      &
      \displaystyle\coprod_{\mathclap{\setQafter\cup \setQafter\alphA}}\langL(\objIn) \arrow[twoheadrightarrow]{d}{\intro*\epiTwoTr} \arrow[rightarrow, bend left]{rd}{\intro*\symbRightTInitTr} & \\
      \langL(\objIn) \arrow[rightarrow]{r}{\intro*\symbLeftQMinTr} \arrow[rightarrow, bend left]{ru}{\intro*\symbLeftQInitTr} \arrow[rightarrow, bend right]{rd}[swap]{\intro*\symbLeftQFinalTr}& \factQT[\setQafter,\setTafter] \arrow[rightarrowtail]{d}{\intro*\monoOneTr} \arrow[rightarrow, shift left]{r}[yshift=0.2ex]{\intro*\symbAQTMinTr} \arrow[rightarrow, shift right]{r}[swap,yshift=-0.4ex]{\intro*\symbEpsilonMinTr=id} &\factQT[\setQafter,\setTafter] \arrow[rightarrow]{r}{\intro*\symbRightTMinTr} \arrow[rightarrowtail]{d}{\intro*\monoTwoTr} &\langL(\objOut)\\
      & \displaystyle\prod_{\mathclap{\setTafter \cup \alphA
          \setTafter}}\langL(\objOut) \arrow[rightarrow, shift
      left]{r}[yshift=0.2ex]{\intro*\symbAQTFinalTr}
      \arrow[rightarrow, shift right]{r}[swap,
      yshift=-0.4ex]{\intro*\symbEpsilonFinalTr} &
      \displaystyle\prod_{\mathclap{\setTafter}}\langL(\objOut)
      \arrow[rightarrow, bend
      right]{ru}[swap]{\intro*\symbRightTFinalTr} &
    \end{tikzcd}
  \end{center}

  We want to show that our canonical arrows between the factorizations
  $\factQT$ and $\factQT[\setQafter,\setTafter]$ cannot be all
  isomorphisms: by contradiction, suppose they are.

  \begin{center}
    \begin{tikzcd}
      {\fact{\setQ,\setTafter}} \arrow[twoheadrightarrow]{d}[sloped, swap]{\cong} \arrow[rightarrowtail]{r}[]{\cong} & {\fact{\setQafter,\setTafter}} \arrow[twoheadrightarrow]{d}[sloped, swap]{\cong} \\
      {\fact{\setQ,\setT}} \arrow[rightarrowtail]{r}[]{\cong} &
      {\fact{\setQafter,\setT}}
    \end{tikzcd}
  \end{center}

  Let's consider the "minimal biautomaton" in
  $\catAuto(\langLQT[\setQafter,\setTafter])$ as a "biautomaton" in
  $\catAuto(\langLQT)$, that is, without the arrows $\symbLeftQMinTr$
  for $q \in \setQafter\setminus\setQ$ and the ones $\symbRightTMinTr$
  for $t \in \setTafter\setminus\setT$.

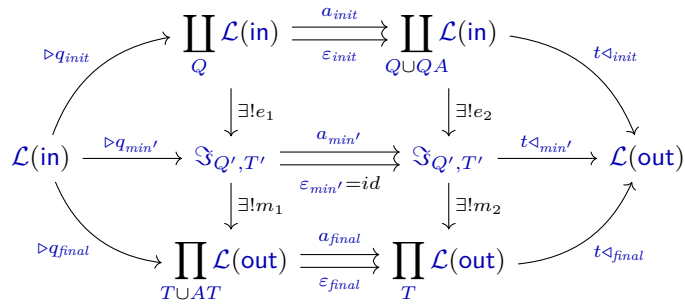
  \begin{center}
    \begin{tikzcd}[column sep=3em]
      & \displaystyle\coprod_{\mathclap{\setQ}}\langL(\objIn)
      \arrow[rightarrow]{d}{\exists! e_1} \arrow[rightarrow, shift
      left]{r}[yshift=0.2ex]{\symbAQTInit} \arrow[rightarrow, shift
      right]{r}[swap, yshift=-0.4ex]{\symbEpsilonInit} &
      \displaystyle\coprod_{\mathclap{\setQ\cup \setQ\alphA}}\langL(\objIn) \arrow[rightarrow]{d}{\exists! e_2} \arrow[rightarrow, bend left]{rd}{\symbRightTInit} & \\
      \langL(\objIn) \arrow[rightarrow]{r}{\symbLeftQMinTr} \arrow[rightarrow, bend left]{ru}{\symbLeftQInit} \arrow[rightarrow, bend right]{rd}[swap]{\symbLeftQFinal}& \factQT[\setQafter,\setTafter] \arrow[rightarrow]{d}{\exists! m_1} \arrow[rightarrow, shift left]{r}[yshift=0.2ex]{\symbAQTMinTr} \arrow[rightarrow, shift right]{r}[swap,yshift=-0.4ex]{\symbEpsilonMinTr=id} &\factQT[\setQafter,\setTafter] \arrow[rightarrow]{r}{\symbRightTMinTr} \arrow[rightarrow]{d}{\exists! m_2} &\langL(\objOut)\\
      & \displaystyle\prod_{\mathclap{\setT \cup \alphA
          \setT}}\langL(\objOut) \arrow[rightarrow, shift
      left]{r}[yshift=0.2ex]{\symbAQTFinal} \arrow[rightarrow, shift
      right]{r}[swap, yshift=-0.4ex]{\symbEpsilonFinal} &
      \displaystyle\prod_{\mathclap{\setT}}\langL(\objOut)
      \arrow[rightarrow, bend right]{ru}[swap]{\symbRightTFinal} &
    \end{tikzcd}
  \end{center}

  It turns out that $e_1,e_2 \in \Epi$ and $m_1,m_2 \in \Mono$: we
  prove just the first fact, the second one being analogous. \AP To
  underline the involved sets, we use $\intro*\res$ and $\intro*\inc$
  to denote respectively the canonical restriction between products
  and the canonical inclusion between coproducts.
  \begin{itemize}
  \item It is easy to check that $e_1$ must be equal to
    $\epiOneTr \circ \inc$ for the universal property of the
    coproduct; as a consequence of one of the diagrams related to our
    factorizations supposed to be isomorphic, we get that
    $e_1 \in \Epi$.
    \begin{center}
      \begin{tikzcd}[column sep=4em]
        \displaystyle\coprod_{\mathclap{\setQ}}\langL(\objIn) \arrow[rightarrow]{d}{\inc[\setQ,\setQafter]} \arrow[twoheadrightarrow]{r}{\epiOne} & \factQT \arrow[rightarrowtail]{r}{\monoTwo} \arrow[rightarrowtail]{d}[sloped,swap]{\cong}& \displaystyle\prod_{\mathclap{\setT}}\langL(\objOut)\\
        \displaystyle\coprod_{\mathclap{\setQafter}}\langL(\objIn) \arrow[transform canvas={xshift=0.3ex},-]{d} \arrow[transform canvas={xshift=-0.3ex},-]{d} \arrow[twoheadrightarrow]{r} & \factQT[\setQafter,\setT] \arrow[twoheadleftarrow]{d}[sloped,swap]{\cong} \arrow[rightarrowtail]{r} & \displaystyle\prod_{\mathclap{\setT}} \langL(\objOut) \arrow[transform canvas={xshift=0.3ex},-]{u} \arrow[transform canvas={xshift=-0.3ex},-]{u} \arrow[leftarrow]{d}{\res[\setTafter,\setT]}\\
        \displaystyle\coprod_{\mathclap{\setQafter}}\langL(\objIn)
        \arrow[twoheadrightarrow]{r}{\epiOneTr} &
        \factQT[\setQafter,\setTafter]\arrow[rightarrowtail]{r}{\monoTwoTr}
        & \displaystyle\prod_{\mathclap{\setTafter}} \langL(\objOut)
      \end{tikzcd}
    \end{center}
  \item It is easy to check that $e_2$ must be equal to
    $\epiTwoTr \circ \inc[\setQ \cup \setQ\alphA,\setQafter \cup
    \setQafter\alphA]$ for the universal property of the coproduct,
    too; the fact that it belongs to $\Epi$ comes from the following
    commuting diagram.
    \begin{center}
      \begin{tikzcd}[column sep=4em]
        \displaystyle\coprod_{\mathclap{\setQ}}\langL(\objIn) \arrow[twoheadrightarrow, bend left]{rr}{\epiOne} \arrow[rightarrow]{r}{\inc[\setQ,\setQ \cup\setQ\alphA]} \arrow[rightarrow]{d}{\inc} &\displaystyle\coprod_{\mathclap{\setQ \cup \setQ\alphA}}\langL(\objIn) \arrow[rightarrow]{d}{\inc[\setQ \cup \setQ\alphA, \setQafter \cup \setQafter\alphA]} \arrow[twoheadrightarrow]{r}{\epiTwo} & \factQT \arrow[rightarrowtail]{r}{\monoTwo} \arrow[rightarrowtail]{d}[swap,sloped]{\cong} & \displaystyle\prod_{\mathclap{\setT}}\langL(\objOut) \arrow[transform canvas={xshift=0.3ex},-]{d} \arrow[transform canvas={xshift=-0.3ex},-]{d}\\
        \displaystyle\coprod_{\mathclap{\setQafter}}\langL(\objIn) \arrow[rightarrow]{r}{\inc[\setQafter,\setQafter \cup\setQafter\alphA]} \arrow[transform canvas={xshift=0.3ex},-]{d} \arrow[transform canvas={xshift=-0.3ex},-]{d} &\displaystyle\coprod_{\mathclap{\setQafter \cup \setQafter\alphA}}\langL(\objIn) \arrow[transform canvas={xshift=0.3ex},-]{d} \arrow[transform canvas={xshift=-0.3ex},-]{d} \arrow[twoheadrightarrow]{r} & \factQT[\setQafter,\setT] \arrow[twoheadleftarrow]{d}[swap,sloped]{\cong} \arrow[rightarrowtail]{r} & \displaystyle\prod_{\mathclap{\setT}} \langL(\objOut) \arrow[leftarrow]{d}{\res} \\
        \displaystyle\coprod_{\mathclap{\setQafter}}\langL(\objIn)
        \arrow[twoheadrightarrow, bend right]{rr}[swap]{\epiOneTr}
        \arrow[rightarrow]{r}{\inc[\setQafter,\setQafter \cup
          \setQafter\alphA]}&\displaystyle\coprod_{\mathclap{\setQafter
            \cup \setQafter\alphA}}\langL(\objIn)
        \arrow[twoheadrightarrow]{r}{\epiTwoTr} &
        \factQT[\setQafter,\setTafter]\arrow[rightarrowtail]{r}{\monoTwoTr}
        & \displaystyle\prod_{\mathclap{\setTafter}} \langL(\objOut)
      \end{tikzcd}
    \end{center}
  \end{itemize}

  So, the two "minimal" objects are isomorphic as "biautomata" in
  $\catAuto(\langLQT)$; as a consequence, the associated "hypothesis
  automata" $\hypQT$ and $\hypQT[\setQafter,\setTafter]$ are
  isomorphic too: therefore, they accept the same language, leading to
  a contradiction.

  In fact, the two "hypothesis automata" associated to the couples
  $(\setQ,\setT)$ and $(\setQafter,\setTafter)$ differ on the
  "counterexample": $\hypQT$ is such that
  $\hypQT\circ i (\symbLeft w \symbRight) \neq \langL(\symbLeft w
  \symbRight)$, whereas $\hypQT[\setQafter,\setTafter]$ is such that
  $\hypQT[\setQafter,\setTafter] \circ i (\symbLeft w \symbRight) =
  \langL(\symbLeft w \symbRight)$, $w$ being the "counterexample"
  whose prefixes are included in $\setQafter$ (see Lemma
  \ref{consistencypreserved}).

  This means that every time a "counterexample" and its prefixes are
  added to $\setQ$, we get a couple of arrows going from the
  factorization related to the previous "hypothesis automaton" to the
  factorization of the following one such that at least one of them is
  not an isomorphism.

\begin{tikzcd}
  \factQT \arrow[rightarrowtail]{r} & \factQT[\setQafter,\setT]
  \arrow[twoheadleftarrow]{r} & \factQT[\setQafter,\setTafter]
\end{tikzcd}

Lemma~\ref{lem:finitechain} states that such a sequence can't be
infinite, so the number of counterexamples with their prefixes that
can be added is finite.
\end{proof}

Proposition~\ref{prop:term-while-loop} together with
Lemma~\ref{algotermination} guarantee that the algorithm terminates
and therefore produces an "automaton" "accepting" the "target
language".

It remains to prove that the produced "automaton" is "minimal": we
make use of the following fact.

\begin{lemma}
  \label{dividemin} Let $\catC$ be a category endowed with an initial
  object $I$, a final object $F$ and a factorization system
  $(\Epi,\Mono)$.
	
  Let $K$ be an "$(\Epi,\Mono)$-noetherian" object dividing the
  "$(\Epi,\Mono)$-minimal" object $\objMinStart$.
	
  $K$ is isomorphic to $\objMinStart$.
\end{lemma}
\begin{proof}
  We know that the "minimal object" $\objMinStart$ in such a category
  is the factorization of the only arrow from the initial to the final
  object and, as a consequence, is defined up to isomorphism (see
  Lemma \ref{minobj}).
	
  Let's suppose that an "$(\Epi,\Mono)$-noetherian" object $K$
  "divides" $\objMinStart$, so there exists an object $\widetilde{K}$
  such that $K$ is a quotient of $\widetilde{K}$, which is a subobject
  of $\objMinStart$.
	
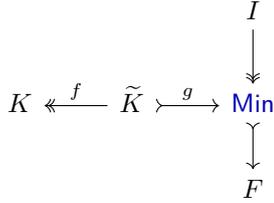
	\begin{tikzcd}
          &&I \arrow[twoheadrightarrow]{d}\\
          K \arrow[twoheadleftarrow]{r}{f} &\widetilde{K} \arrow[rightarrowtail]{r}{g} &\objMinStart\\
          &&F \arrow[leftarrowtail]{u}
	\end{tikzcd}

	Consider the only arrow from $K$ to $F$ together with its
        factorization $K'$, $f'$ being the morphism in $\Epi$ of the
        $(\Epi,\Mono)$-factorization of such an arrow.
	
	
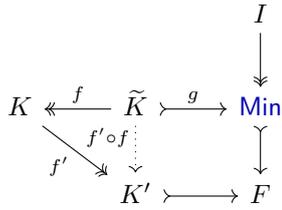
\begin{tikzcd}
          &&I \arrow[twoheadrightarrow]{d}\\
          K \arrow[twoheadleftarrow]{r}{f} \arrow[twoheadrightarrow]{rd}[swap]{f'} &\widetilde{K} \arrow[rightarrowtail]{r}{g} \arrow[dotted]{d}[swap, near start]{f' \circ f} &\objMinStart\\
          &K' \arrow[rightarrowtail]{r}&F \arrow[leftarrowtail]{u}
	\end{tikzcd}

	$f' \circ f$ clearly belongs to $\Epi$, but also belongs to
        $\Mono$, as the square on the right commutes ($F$ is a final
        object): as a consequence, $f' \circ f$ is an isomorphism,
        hence we may suppose that $K'=\widetilde{K}$ and
        $f' \circ f=id_{\widetilde{K}}$.
	
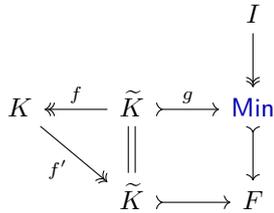
	\begin{tikzcd}
          &&I \arrow[twoheadrightarrow]{d}\\
          K \arrow[twoheadleftarrow]{r}{f} \arrow[twoheadrightarrow]{rd}[swap]{f'} &\widetilde{K} \arrow[transform canvas={xshift=0.3ex},-]{d} \arrow[transform canvas={xshift=-0.3ex},-]{d} \arrow[rightarrowtail]{r}{g} &\objMinStart\\
          &\widetilde{K} \arrow[rightarrowtail]{r}&F
          \arrow[leftarrowtail]{u}
	\end{tikzcd}
	
	Now observe that the following diagram commutes.
	
	
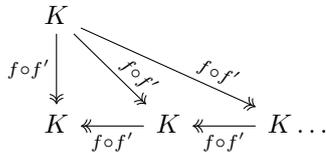
\begin{tikzcd}[row sep=large]
          K \arrow[twoheadrightarrow]{d}[swap]{f \circ f'}
          \arrow[twoheadrightarrow,sloped,near end]{dr}{f \circ f'}
          \arrow[twoheadrightarrow,sloped,near end]{drr}{f \circ f'}\\
          K \arrow[twoheadleftarrow]{r}[swap]{f \circ f'} & K
          \arrow[twoheadleftarrow]{r}[swap]{f \circ f'} & K \ldots
	\end{tikzcd}
	
	By virtue of the "$(\Epi,\Mono)$-noetherianity" of $K$, $f
        \circ f'$ must be an isomorphism; in these conditions, it is
        easy to check that $f \circ f'=id_{K}$, so $f$ is an
        isomorphism.
	
	An analogous reasoning with arrows in $\Mono$ shows that $g$
        is an isomorphism too, concluding the proof.
      \end{proof}

\begin{lemma}
  \label{lem:correctness-algo}
  The "automaton" produced by Algorithm~\ref{algorithm:main} is the
  "minimal" one $\objMin(\langL)$.
\end{lemma}

\begin{proof}
  Let $(\setQ,\setT)$ be the last couple produced by the algorithm.
	
  In $\catAuto(\langL)$, the "minimal" $\objMin(\langL)$ divides any
  other object, the "hypothesis automaton" $\hypQT$ too; so,
  $\objMin(\langL)$ is a quotient of a certain "automaton" $\autA$
  being subobject of $\hypQT$.
	
  Now, we can consider all these automata as particular "biautomata",
  taking their precomposition with the functor $\iStar$ (see
  Lemma~\ref{lem:input-cats}).
	
  In $\catAuto(\langLQT)$, it is easy to check that
  $\objMin(\langL) \circ \iStar$ is "$(\Epi,\Mono)$-noetherian" and
  divides $\hypQT \circ \iStar$ by means of $\autA \circ \iStar$
  itself.
	
  In addition, $\hypQT \circ \iStar$ is clearly the "minimal
  biautomaton".
	
  Lemma~\ref{dividemin} guarantees that $\hypQT \circ \iStar$ and
  $\objMin(\langL) \circ \iStar$ are isomorphic, so they are also
  isomorphic as "hypothesis automata", that is, as objects in
  $\catAuto(\langL)$.
\end{proof}

\section{The optimized algorithm}
\label{sec:opt-algorithm}

The aim of this section is to prove
Theorem~\ref{thm:optimized-algorithm}, that is, the termination and
correctness of the following optimized algorithm. Recall that the
canonical morphisms between the factorizations appearing on
lines~\ref{line5opt} and~\ref{line8opt} of
Algorithm~\ref{algorithm:optimized} are as defined in
Lemma~\ref{lem:factQT-properties}.

\begin{algorithm}
  \caption{The optimized \FunL learning algorithm}
  \label{algorithm:optimized}
  \SetKwInOut{Input}{input}\SetKwInOut{Output}{output}
  \Input{minimally adequate teacher of the "target language" $\langL$}
  \Output{$\objMin(\langL)$}
  $\setQ:=\setT:=\left \{ {\varepsilon} \right \}$
	
  \Repeat{the answer is yes}{ \While{$(\setQ,\setT)$ is not
      "$\langL$-automatable"}{ \If{$\symbEpsilonMin \notin \Epi$}{add
        $qa \in \setQ\alphA$ s.t.
        $\factQT(\langL) \rightarrowtail \factQT[\setQ \cup \left \{qa
        \right \},\setT](\langL)$ is not an isomorphism to
        $\setQ$ \label{line5opt}}
      \If{$\symbEpsilonMin \notin \Mono$}{add $at \in \alphA \setT$
        s.t.
        $\factQT(\langL) \twoheadleftarrow \factQT[\setQ,\setT \cup
        \left \{at \right \}](\langL)$ is not an isomorphism to
        $\setT$ \label{line8opt}}} ask an "equivalence query" for the
    "hypothesis automaton" $\hypQT$
		
    \If{the answer is no}{add the provided "counterexample" and all
      its prefixes to $\setQ$} }
	
  \Return{$\hypQT$}
\end{algorithm}

In Proposition~\ref{prop:term-while-loop}, we have proved the
termination of the while cycle as a consequence of the finiteness of
the chain of arrows not being isomorphisms among the factorizations
$\factQT$ of the different couples.
	
Therefore, as regards the correctness of the optimized algorithm, the
only fact that we have to check is that such a $qa \in \setQ\alphA$ if
$\symbEpsilonMin \notin \Epi$ and such an $at \in \alphA\setT$ if
$\symbEpsilonMin \notin \Mono $ exist.
	
We prove it in the following two lemmas, one being the dual of the
other.
	
\begin{lemma}
  Let $\symbEpsilonMin \notin \Mono$.  There exists a word
  $at \in \alphA\setT$ s.t.
  $\factQT \twoheadleftarrow \factQT[\setQ, \setT \cup \left \{at
  \right \}]$ is not an isomorphism.
\end{lemma}

\begin{proof}
  We have proved in Lemma~\ref{lem:epi-mono-fact-of-epsilon} that
  $\eWhile \in \Epi\setminus\Mono$, \AP $\intro*\eWhile$ being
  obtained as follows.
  \begin{center}
    \begin{tikzcd}
      \displaystyle\coprod_{\mathclap{\setQ}}\langL(\objIn) \arrow[transform canvas={xshift=0.3ex},-]{d} \arrow[transform canvas={xshift=-0.3ex},-]{d} \arrow[twoheadrightarrow]{r} & \factQT \arrow[rightarrowtail]{r} \arrow[twoheadleftarrow]{d}{\exists!\eWhile} & \displaystyle\prod_{\mathclap{\setT}}\langL(\objOut) \arrow[leftarrow]{d}{\res[\setT \cup \alphA\setT,\setT]}\\
      \displaystyle\coprod_{\mathclap{\setQ}}\langL(\objIn)
      \arrow[twoheadrightarrow]{r} & \factQT[\setQ,\setT \cup \alphA
      \setT] \arrow[rightarrowtail]{r} &
      \displaystyle\prod_{\mathclap{\setT \cup \alphA \setT}}
      \langL(\objOut)
    \end{tikzcd}
  \end{center}
  By contradiction, let's suppose that such a word $at$ does not
  exist.  As a consequence, we have the following situation for all
  $at \in \alphA\setT$.
  \begin{center}
    \begin{tikzcd}
      \displaystyle\coprod_{\mathclap{\setQ}}\langL(\objIn) \arrow[twoheadrightarrow]{r} & \factQT \arrow[rightarrowtail]{r} \arrow[tail]{rd} & \displaystyle\prod_{\mathclap{\setT}}\langL(\objOut) \arrow[leftarrow]{d}{\res[\setT \cup \left \{at \right \},\setT]}\\
      & \factQT[\setQ,\setT \cup \left\{at\right\}] \ar[r,
      tail]\ar[u,dotted,two heads,tail,"\cong" description] &
      \displaystyle\prod_{\mathclap{\setT \cup \left \{at \right \}}}
      \langL(\objOut)
    \end{tikzcd}
  \end{center}
  Using the diagonal arrows in the previous diagram, which are all
  equal when restricted to $\prod_{\setT}\langL(\objOut)$, we obtain
  an arrow from $\factQT$ to the bigger product that makes the
  following diagram commute.
  \begin{center}
    \begin{tikzcd}
      \displaystyle\coprod_{\mathclap{\setQ}}\langL(\objIn) \arrow[twoheadrightarrow]{r} & \factQT \arrow[twoheadrightarrow,swap]{dd}{e'} \arrow[rightarrow,bend right]{rdd} \arrow[rightarrowtail]{r} \arrow[rightarrowtail]{rd} & \displaystyle\prod_{\mathclap{\setT}}\langL(\objOut) \arrow[leftarrow]{d}{\res[\setT \cup \left \{at \right \},\setT]}\\
      & & \displaystyle\prod_{\mathclap{\setT \cup \left \{at \right \}}} \langL(\objOut)\\
      & \factQT[\setQ,\setT \cup \alphA \setT]
      \arrow[rightarrowtail]{r} & \displaystyle\prod_{\mathclap{\setT
          \cup \alphA\setT}}\langL(\objOut)
      \arrow[rightarrow]{u}[swap]{\res[\setT \cup \alphA\setT, \setT
        \cup \left \{at \right \}]}
    \end{tikzcd}
  \end{center}
  Consider the $(\Epi,\Mono)$-factorization of this map, with $e'$
  denoting the arrow in $\Epi$.  It is easy to see that $e'$ is an
  inverse of $\eWhile$, as their composition is uniquely determined by
  the universal property of the factorization system.

  This is a contradiction, as $\eWhile$ is not an isomorphism.
\end{proof}

\begin{lemma}
  Let $\symbEpsilonMin \notin \Epi$.  There exists a word
  $qa \in \setQ\alphA$ s.t.
  $\factQT \rightarrowtail \factQT[\setQ \cup \left \{qa \right
  \},\setT]$ is not an isomorphism.
\end{lemma}

\begin{proof}
  We have proved in Lemma~\ref{lem:epi-mono-fact-of-epsilon} that
  $\mWhile \in \Mono\setminus\Epi$, \AP $\intro*\mWhile$ being
  obtained as follows.
  \begin{center}
    \begin{tikzcd}
      \displaystyle\coprod_{\mathclap{\setQ}}\langL(\objIn) \arrow[twoheadrightarrow]{r} & \factQT \arrow[rightarrowtail]{r} \arrow[rightarrowtail]{d}{\exists!\mWhile} & \displaystyle\prod_{\mathclap{\setT}}\langL(\objOut) \arrow[transform canvas={xshift=0.3ex},-]{d} \arrow[transform canvas={xshift=-0.3ex},-]{d}\\
      \displaystyle\coprod_{\mathclap{\setQ \cup
          \setQ\alphA}}\langL(\objIn)
      \arrow[leftarrow]{u}{\inc[\setQ,\setQ \cup \setQ\alphA]}
      \arrow[twoheadrightarrow]{r} & \factQT[\setQ \cup
      \setQ\alphA,\setT] \arrow[rightarrowtail]{r} &
      \displaystyle\prod_{\mathclap{\setT}} \langL(\objOut)
    \end{tikzcd}
  \end{center}

  By contradiction, let's suppose that such a word does not exist.  As
  a consequence, we have the following situation for all
  $qa \in \setQ\alphA$.
  \begin{center}
    \begin{tikzcd}
      \displaystyle\coprod_{\mathclap{\setQ}}\langL(\objIn) \arrow[twoheadrightarrow]{r} & \factQT \arrow[rightarrowtail]{r} & \displaystyle\prod_{\mathclap{\setT}}\langL(\objOut)\\
      \displaystyle\coprod_{\mathclap{\setQ \cup \left \{ qa
          \right\}}}\langL(\objIn)
      \arrow[leftarrow]{u}{\inc[\setQ,\setQ \cup \left \{ qa
        \right\}]} \arrow[twoheadrightarrow]{ru}
    \end{tikzcd}
  \end{center}
  We have all we need to build up another arrow starting from the
  bigger coproduct and making the following commute, with $m'$
  denoting the morphism in $\Mono$ of the $(\Epi,\Mono)$-factorization
  of such an arrow.
  \begin{center}
    \begin{tikzcd}
      \displaystyle\coprod_{\mathclap{\setQ}}\langL(\objIn) \arrow[twoheadrightarrow]{r} & \factQT \arrow[rightarrowtail]{r} & \displaystyle\prod_{\mathclap{\setT}}\langL(\objOut)\\
      \displaystyle\coprod_{\mathclap{\setQ \cup \left \{ qa \right\}}}\langL(\objIn) \arrow[leftarrow]{u}{\inc[\setQ,\setQ \cup \left \{ qa \right\}]} \arrow[twoheadrightarrow]{ru}\\
      \displaystyle\coprod_{\mathclap{\setQ \cup
          \setQ\alphA}}\langL(\objIn) \arrow[leftarrow]{u}{\inc[\setQ
        \cup \left \{qa \right \},\setQ \cup \setQ\alphA]}
      \arrow[rightarrow, bend right, xshift=-1ex]{ruu}
      \arrow[twoheadrightarrow]{r} & \factQT[\setQ \cup
      \setQ\alphA,\setT] \arrow[rightarrowtail]{uu}[swap]{m'}
    \end{tikzcd}
  \end{center}
  It is easy to see that $m'$ is an inverse of $\mWhile$, as their
  composition is uniquely determined by the universal property of the
  factorization system.
	
  This is a contradiction, as $\mWhile$ is not an isomorphism.
\end{proof}
\end{document}